%% file: main.tex
\documentclass[conference]{IEEEtran}
\usepackage{cite}
\usepackage{amsmath,amssymb,amsfonts,amsthm}
\usepackage{graphicx}
\usepackage{textcomp}
\usepackage{xcolor}
\def\BibTeX{{\rm B\kern-.05em{\sc i\kern-.025em b}\kern-.08em
    T\kern-.1667em\lower.7ex\hbox{E}\kern-.125emX}}
\usepackage{subcaption}
\usepackage{algorithm}
\usepackage{algorithmic}
\usepackage{multirow}
\newtheorem{theorem}{Theorem}

\begin{document}

%\title{Constrained Reinforcement Learning for Network Slicing}
\title{CLARA: A Constrained Reinforcement Learning Based Resource Allocation Framework for \\ Network Slicing}

% \author{Yongshuai Liu}
% \email{yshliu@ucdavis.edu}
% \author{Jiaxin Ding}
% %\authornotemark[1]
% \affiliation{%
%%%ACM author
%\author{Yongshuai Liu, Jiaxin Ding, Xin Liu}
%\affiliation{\institution{University of California, Davis}}
%\email{{yshliu,jxding,xinliu}@ucdavis.edu}

%%IEEE author
\author{
\IEEEauthorblockN{Yongshuai Liu\IEEEauthorrefmark{1}, Jiaxin Ding\IEEEauthorrefmark{2}, Zhi-Li Zhang\IEEEauthorrefmark{3}, Xin Liu\IEEEauthorrefmark{1}}
    \IEEEauthorblockA{\IEEEauthorrefmark{1}Department of Computer Science, University of California, Davis}
    \IEEEauthorblockA{\IEEEauthorrefmark{2}John Hopcroft Center for Computer Science, Shanghai Jiao Tong University 
    \IEEEauthorblockA{\IEEEauthorrefmark{3}Department of Computer Science and Engineering, University of Minnesota}
    Emails: yshliu@ucdavis.edu, jiaxinding@sjtu.edu.cn, zhzhang@cs.umn.edu, xinliu@ucdavis.edu}
}
%\author{\IEEEauthorblockN{1\textsuperscript{st} Yongshuai Liu}
%\IEEEauthorblockA{\textit{Computer Science Department, University of California, Davis} \\
%\and


\maketitle
\begin{abstract}
As mobile networks proliferate, we are experiencing a strong diversification of services, which requires greater flexibility from the existing network. Network slicing is proposed as a promising solution for resource utilization in 5G and future networks to address this dire need. In network slicing, dynamic resource orchestration and network slice management are crucial for maximizing resource utilization. Unfortunately, this process is too complex for traditional approaches to be effective due to a lack of accurate models and dynamic hidden structures. We formulate the problem as a Constrained Markov Decision Process (CMDP) without knowing models and hidden structures. Additionally, we propose to solve the problem using CLARA, a Constrained reinforcement LeArning based Resource Allocation algorithm. In particular, we analyze cumulative and instantaneous constraints using adaptive interior-point policy optimization and projection layer, respectively. Evaluations show that CLARA clearly outperforms baselines in resource allocation with service demand guarantees. 
\end{abstract}




% \keywords{Networks Slicing, 5G, Reinforcement Learning, Constrained Optimization}


\input{introduction.tex}

%\input{why.tex}

\input{system.tex}

\input{formulation.tex}

\input{preliminary.tex}

\input{method.tex}

\input{experiment.tex}
\input{related.tex}

\section{Conclusion}
This work focuses on network slicing with constraints. 
Because of the service diversity, complexity, and hidden structures, prior approaches fail to satisfy the cumulative and instantaneous service constraints. 
To address this challenge,  we formulate the network slicing problem as a Constrained Markov Decision Process (CMDP) and solve it with CLARA, which is a reinforcement learning based resource allocation framework for network slicing with constraints. We evaluate CLARA on the radio resource allocation scenario to illustrate the details of the proposed approach and its advantage. 
Our evaluation results show that constrained reinforcement learning can solve network slicing problems effectively. 
Much future work exists, including stronger theoretical bounds, improved sample efficiency, testbed development, as well as  more dynamic real-world evaluations. 

% \section*{Acknowledgment}
% The work was partially supported by NSF through grants IIS-1838207, CNS 1901218, OIA-2040680,  OIA-2134901 and USDA-020-67021-32855. J. Ding would like to acknowledge supports from Shanghai Sailing Program 20YF1421300. 

%\bibliographystyle{ACM-Reference-Format}
\bibliographystyle{IEEEtran}
\bibliography{references}
%\balance
%\appendix

\end{document}

%% file: introduction.tex
\section{Introduction}
As mobile networks proliferate, we face a broad variety of services, including autonomous driving, industry 4.0, virtual/augmented reality, and the Internet of Things (IoT), etc. These services are characterized by heterogeneous performance, functional, security and operational requirements~\cite{lu2015safeguard,qu2015improved,lu2017sense,liu2018less}, which demand the network to embed more flexibility. Because of the lack of flexibility in existing networks, independent initiatives of dedicated infrastructure solutions are deployed: 3GPP has developed an IoT specific MAC that can co-exist with general purpose MAC~\cite{3gpp2015cellular}; the industry has deployed proprietary architectures for extreme reliability~\cite{li2017review}. Such monolithic vertical developments are clearly highly expensive and inefficient. 
%To satisfy this daring need for service diversification, \ys{the policy-based network\cite{strassner2004policy} is embraced as a promising solution derived from IETF work in Resource Allocation Protocol working group and in Policy Framework working group.}

% In 5G and future networks, network slicing~\cite{foukas2017network}, enabled by network function virtualization (NFV) and software defined networking (SDN), is embraced as a promising solution for flexible resource provisioning. Network slicing creates multiple virtual network instances, named network slices, with service isolation and guarantees on top of a common physical infrastructure.

Network slicing~\cite{foukas2017network} is a promising solution for flexible resource provisioning in 5G and future networks. It creates multiple virtual instances, named network slices, over physical infrastructures enabled by network function virtualization (NFV) and software-defined networking (SDN).
Different slices can accommodate different service demands, such as 
1) Mobile broadband (high throughput/low latency); 
2) massive IoT communication (massive connection/bursty traffic); 
3) video surveillance system (high edge processing requirement); 
4) autonomous driving (ultra  reliability/low latency); 
5) VR/AR/real-time gaming (low latency/ultra high throughput); 
6) industry 4.0 automation (ultra reliability/security). 
% Network slicing can be performed at different parts of the network, including core cloud, edge cloud, radio resource management, RAN processing, spectrum, and radio frontend, etc.~\cite{alliance20155g, foukas2017network, nikaein2015network}, to manage and allocate different resources to satisfy the different service demands.  
Network slicing can take place at different places on the network, including core cloud, edge cloud, radio resource management, RAN processing, spectrum, and radio frontend.~\cite{alliance20155g, foukas2017network, nikaein2015network}. Generally, it aims to manage and allocate different resources to meet various service demands. 
% It's a time-slotted system where network slicing decisions are made at the beginning of each time slot.
% In essence, network slicing is a generalized resource allocation problem over heterogeneous resources to meet the service demands, in compliance with the complex network dynamics, in the long run. Such dynamic orchestration of network slices is essential for resource efficiency~\cite{zhang2017network}. 

Network slicing is essentially a generalized resource allocation problem that manages heterogeneous resources to meet service demands in accordance with the dynamics of the complex network. It is essential to orchestrate network slices dynamically to maximize resource efficiency~\cite{zhang2017network}. 
% However, the resource allocation in network slicing is a highly complicated problem that the existing traditional approaches cannot solve effectively and efficiently. 
Nevertheless, resource allocation is a highly complex issue in network slicing, which traditional approaches cannot resolve efficiently or effectively. 

% traditional optimization approaches require accurate mathematical models with parameters known, which is often difficult to achieve in practice, especially with the increasing complexity, scale, and service diversity of the 5G and future networks. 
First, traditional optimization methods only incorporate models with precisely known parameters, which is often impossible in practice, especially as 5G and future networks become increasingly complex, large, and diverse in service. 
% Constraints from the physical systems and service demands are prevalent and complex, such as latency requirement and service level agreement, which further increases the difficulty.
It is further complicated by the constraints from the physical system and service requirements, such as latency requirements and service level agreements.
% , let alone obtaining a closed-form expression. 
 The network environment is affected by the location, geographical properties (e.g., tall buildings, hills, highways, etc), mobility, etc., all difficult to measure and model. 
 The precise network traffic of a time slot is not available, when the network slicing decision is made at the beginning of the time slot. 
 Because of these reasons, it is very difficult to obtain accurate model parameters in real networks, on which traditional optimization approaches rely.
%  Because of these reasons, it is very difficult to accurately model network slicing in real networks and obtain accurate model parameters, on which traditional optimization approaches rely.
 %Furthermore, network slicing decisions are made at the beginning of a time slot, without precise information on the traffic demand during the whole time slot. 

% Additionally, traditional methods do not adapt to epistemic uncertainty, exhibited as hidden structures in networks, due to a lack of knowledge and subsequent ability to explore and learn from the studied system. For example, a user who experiences a poor quality of service, may decide not to use the service again or at a reduced frequency. Such hidden structures can significantly affect user experience and network performance, but are usually not directly observed/modeled. 
Additionally, traditional methods do not accommodate epistemic uncertainty, manifested as hidden structures in networks, resulting from a lack of knowledge and subsequent ability to explore and learn from the network. For example, a user who experiences a poor quality of service may decide not to use the service again or at a reduced frequency. Such hidden structures can significantly affect user experience and network performance but are usually not directly observed/modeled. 
%In such cases, exploration in learning-based algorithms is inherently beneficial. 

Facing these challenges, learning-based approaches are beneficial for addressing these challenges because they have the ability to explore and learn from a network without assuming the prior knowledge of accurate models. 
The industry has recognized machine learning (ML) as a core technology for future telecommunication networks, including 5G and beyond~\cite{itfuture}. 
%However, machine learning is not a cure-all. 
%Among existing learning-based network research, 
A growing number of recent research studies have demonstrated significant performance improvement using learning-based networks, e.g.,~\cite{mao2017neural,uzakgider2015learning,mao2016resource,ChuaiInfocom2019,BaoBigData2016,zhang2019macs,sengupta2018hotdash}.  
% However, few previous work has analyzed the resource allocation problem with the constraints imposed by the service requirements, which is the crucial for network slicing.  
However, they either only consider one step of optimization (e.g. multi-armed bandit~\cite{ChuaiInfocom2019}) or do not analyze the resource allocation problem with constraints imposed by the service requirements, which is crucial for network slicing.
%, while others barely outperform current algorithms. 
%We discuss the suitability of learning-based approaches in network slicing in Sec.~\ref{sec:why}, and our numerical results demonstrate strong performance benefits compared to traditional baselines.
%Specifically, 
%we apply constrained RL on a radio resource slicing scenario including the challenges listed above, which will be discussed in detail in Sec \ref{sec:scenario}. The evaluation results show that constrained RL can address these challenges effectively, and outperform the traditional schemes.
%In most, if not all, network resource allocation problems, in addition to maximizing the long term reward (e.g. total throughput and user Quality of Experience), we need to satisfy a number of constraints. 

In this work, we propose CLARA, a Constrained reinforcement LeArning based Resource Allocation framework for network slicing. 
Thanks to NFV, we can focus our resource allocation decisions on the virtualized resources. 
We first model the problem as a Constrained Markov Decision Process (CMDP) without knowing prior knowledge. The objective is to maximize a reward in a long run, e.g. throughput over time. At the same time, it is subject to a number of constraints. 
The constraints abstracted from system capacity limits and service requirements are formulated as two types of constraints, cumulative and instantaneous constraints.  
A cumulative constraint requires that the sum of a quality is within a certain limit, e.g., outage probability or average throughput, etc, while an instantaneous constraint requires that the quality needs to satisfy a condition in each time slot, e.g. resource limits and service latency requirement, etc. 
Instantaneous constraints can be further divided into explicit and implicit instantaneous constraints. An explicit constraint has a closed-form expression that can be numerically checked, e.g., transmission power and spectrum available, etc. 
%Available network resources are in a general explain constraints.
An implicit  constraint does not have an accurate closed-form formulation due to the complexity of the system, e.g., latency and interference level, etc.

%Existing solutions based on theoretical modeling suffer from the lack of accurate models and the existence of hidden structures which will be discussed in more detail in Sec.~\ref{sec:why}. Furthermore, they didn't consider constraints in network.

%Network slicing is highly complex, involving various constraints, multiple layers/operations, numerous users, diverse application types and behavior. To address these difficulties, reinforcement learning (RL) is a promising solution. In this work, 
We develop efficient reinforcement learning algorithms  for network slicing under both cumulative and instantaneous constraints with theoretical analysis on the performance bound. 
% To the best of our knowledge, we are the first one to apply constrained RL for network slicing under constraints. 
Specifically, to deal with cumulative constraints, we propose our adaptive constrained RL algorithm improved over  Interior-point Policy Optimization (IPO)~\cite{liu2020ipo}. 
%We demonstrate that our algorithm can easily handle multiple constraints, a big advantage compared to prior work. 
% Intuitively, we want a penalty function such that 1) if a constraint is satisfied, the penalty added to the reward function is zero, and 2) if the constraint is violated, the penalty added goes to negative infinity.
For instantaneous constraints, we project a resource allocation decision generated by the reinforcement learning algorithm to its nearest feasible decision 
%the feasible slices to the nearest feasible slices by adding a safety layer 
at the end of policy neural network~\cite{dalal2018safe,bhatia2019resource}. 
%
%In summary, our contributions are as follows:
%\begin{itemize}
%    \item We propose CLARA, a constrained reinforcement learning based resource allocation framework. 
%    \item We provide theoretical guarantees for the performance improvement when updating the policy during training. 
%    \item We further develop a practical implementation that updates hyperparameters adaptively, which satisfies both cumulative and instantaneous constraints, and speeds up convergence.   
%    \item We present a case study of  radio access network slicing and conduct experiments to demonstrate the effectiveness of CLARA. %compared with baselines. %Furthermore, it can be extended to handle multiply constraints.
%%    model based solutions. The results shows that, constrained RL solves network slicing problem sufficiently and it can be easily extended to problems with multiple constraints.
%\end{itemize}
%
%%%%% The rest of paper is organized as follows. 

%% file: system.tex
\section{Network Slicing}
\subsection{System Description}
\label{Sec:sys}
%\begin{figure}
%\includegraphics[width=\columnwidth]{figures/NS.pdf}
%\caption{}
%\end{figure}
%
%\begin{figure}
%\includegraphics[width=\columnwidth]{figures/VNF.pdf}
%\caption{}
%\end{figure}

\begin{figure*}[tbh!]
     \centering
          \hfill
     \begin{subfigure}[t]{0.32\textwidth}
         \centering
         \includegraphics[width=\textwidth]{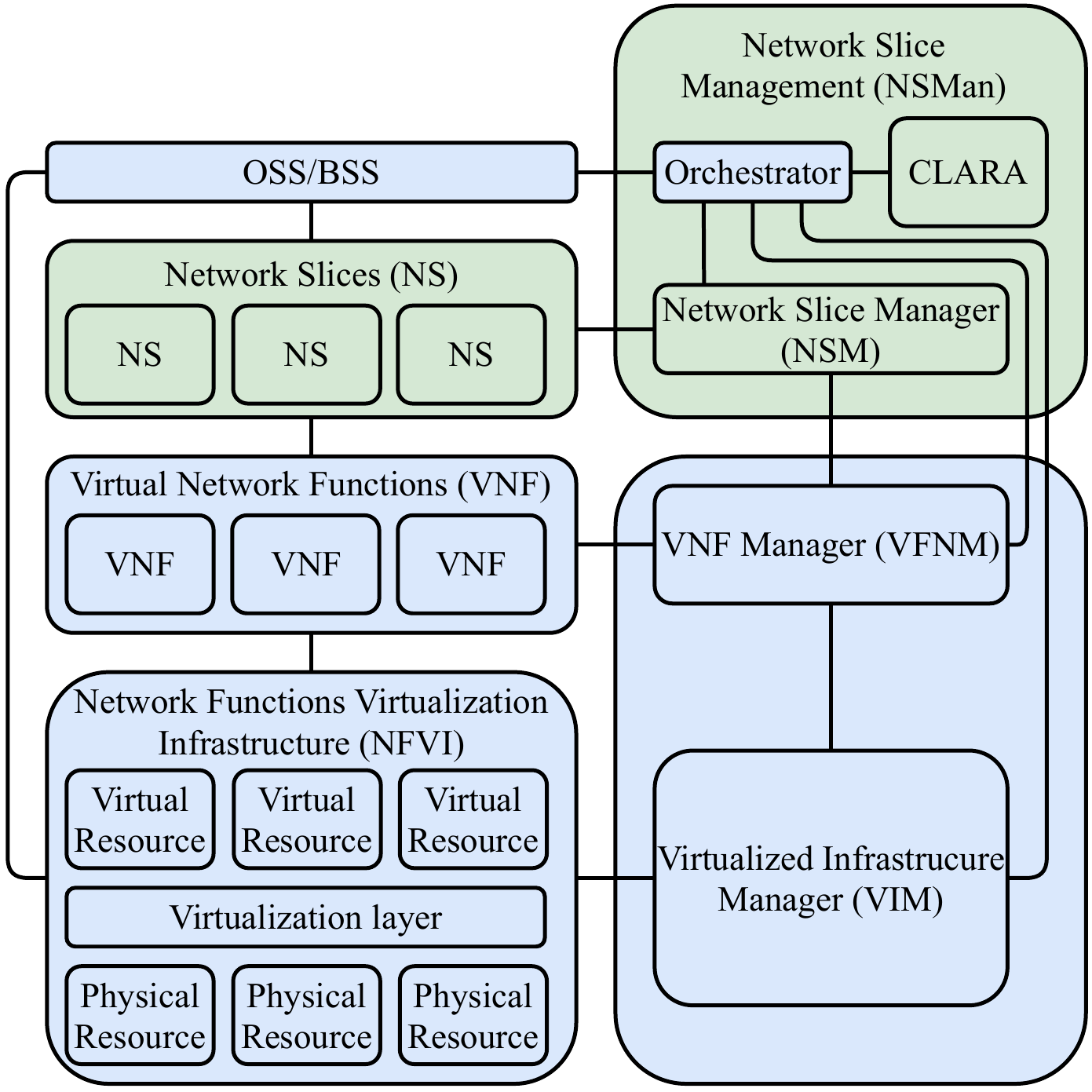}
         \caption{The network slicing architecture with CLARA}
         \label{fig:VNF}
     \end{subfigure}
     \hfill
     \begin{subfigure}[t]{0.63\textwidth}
         \centering
         \includegraphics[width=\textwidth]{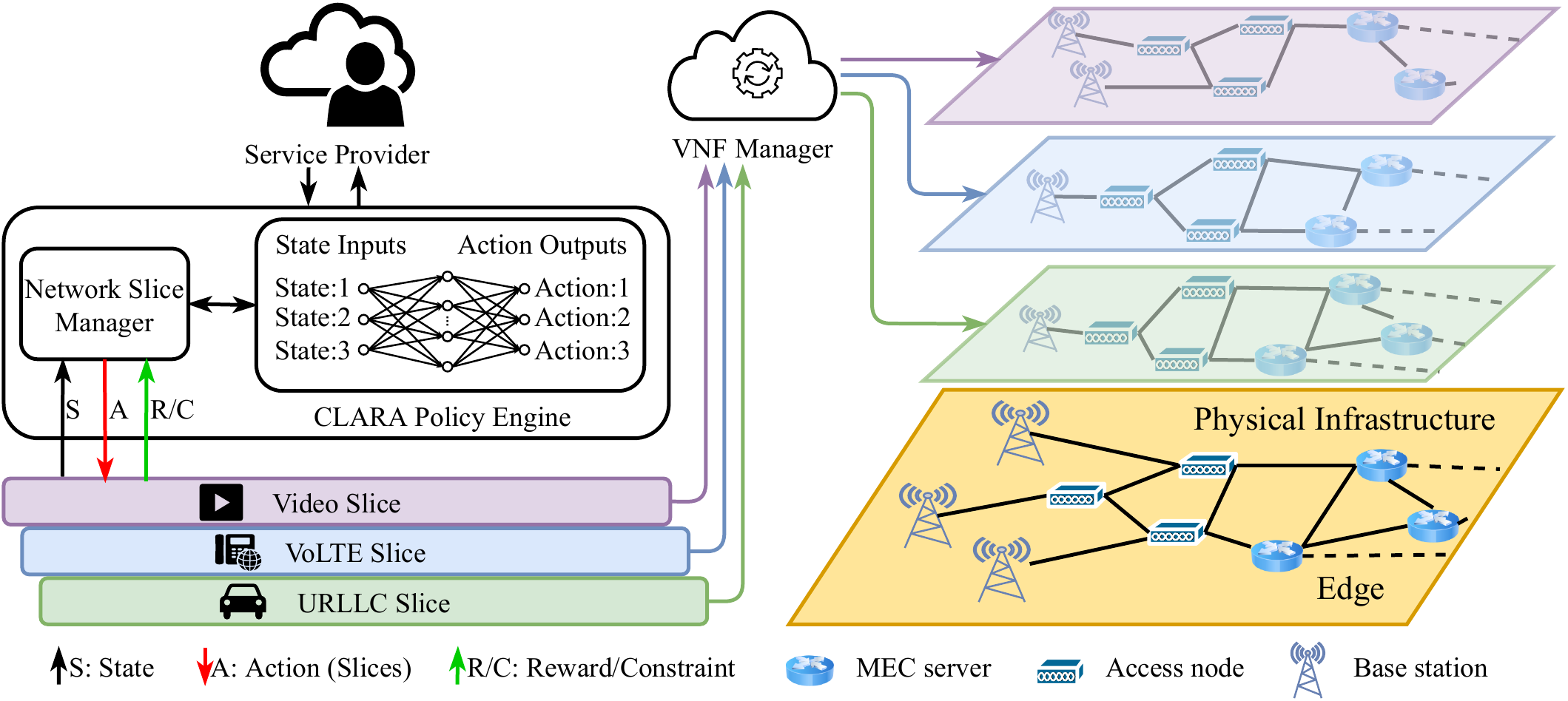}
         \caption{A use case of  network slicing with CLARA}
         \label{fig:NS}
     \end{subfigure}
    %  \vspace{-2mm}
     \caption{Network slicing with CLARA}
    %  \vspace{-4mm}
     \end{figure*}
We briefly describe the network slicing architecture with CLARA for resource allocation management, as shown in Fig.~\ref{fig:VNF}. It is developed from the ETSI reference model~\cite{etsi2015network} colored in blue, and our contribution is highlighted with green. 

From bottom to top, Fig.~\ref{fig:VNF} first shows that Virtualized Infrastructure Management (VIM) is in charge of the Network Functions Virtualization Infrastructure (NFVI),
which converts the physical resources to virtual resources, through the Virtualization layer. 
Thereafter, Virtual Network Function (VNF) consists of the virtual resources, which is monitored, configured, and controlled by the VNF manager. 
These VNFs providing different service functionalities finally make up the Network Slice (NS). 
Because of NFV, we are able to allocate resources according to virtual resources. 
The Network Slicing Manager (NSM) is responsible for the initialization, configuration, and managing the life cycles of the network slices. 
The Orchestrator automatizes the management of the elements on the lower levels, communicates and coordinates all those control managers to schedule tasks and update states of the system. 
The Orchestrator is linked to Operations Support System/Business Support System (OSS/BSS) which provides the services. 
Our CLARA, a policy neural network trained with constrained RL algorithms, is in charge of providing solutions for the resource allocation on network slices to the NSM and the Orchestrator, and receive feedbacks to improve the policy on further decisions.  

% We can go through CLARA in details, in Fig.~\ref{fig:NS}. 
Fig.~\ref{fig:NS} shows CLARA in details.
The service providers submit the service requirements to CLARA. 
The Network Slice Manager monitors the current system state and sends it to CLARA. 
Based on the system state and the service requirements, CLARA proposes the resource allocation plan to the NSM. 
The NSM configure the network slices with the proposed plan. 
This plan is passed to the VNF manager and lower level components to further map the virtual resources to physical resources. 
For each configuration in a decision time slot, the NSM monitors the system and network slices to measure the rewards(constraints) collected and send the rewards(constraints) to CLARA for further policy improvement. 

\subsection{Radio Access Network Slicing}\label{sec:scenario} 
To be clear, we consider how to apply CLARA for radio access scenario with hidden dynamics. Specifically, give a list of slices $1,...,N$ sharing bandwidth $B$. CLARA learns a policy to give a bandwidth allocation $(b_1,...,b_N)$ to maximize the throughput over time while satisfying certain constraints.

% \begin{table*}[t]
% 	\centering
% 	\caption{Parameter for different types of user}
% 	\begin{tabular}{|c c c c|} 
% 		\hline
% 		Distribution& Initial number of users&Inter-Arrival Time&Packet Size\\
% 		 \hline
% 		 \multirow{2}{*}{Video}&\multirow{2}{*}{Poisson [Mean=50]}&Pareto [Exponential Para = 1.2, &Truncated Pareto [Exponential Para = 1.2, \\
% 		 &&Mean= 6 ms, Max = 12.5 ms]& Mean= 100 Byte, Max = 250Byte]\\
% 		 \hline
% 		 VoLTE& Poisson [Mean=50]  &Uniform [Min = 0, Max =160ms] &Constant (40 Byte) \\
% 		 \hline
% 		 \multirow{2}{*}{URLLC}&\multirow{2}{*}{Poisson [Mean=10]}&Truncated Exponential&Truncated Lognormal [Mean = 2 MB, \\
% 		 &&[Mean = 180ms]&Standard Deviation = 0.722 MB, Maximum =5 MB]\\
% 		\hline
% 	\end{tabular}
% 	\label{table:paremeter}
% 	\vspace{-4mm}
% \end{table*}

We simulate a scenario containing a Base Station (BS) with three types of services (i.e., Video, VoLTE, URLLC). 
Each service has a random number (max 100 for each type) of users located surrounding the BS. The users arrive according to Poisson distributions (initial mean value 50, 50, and 10 respectively).  
User requests of VoIP and video services take the parameter settings of VoLTE and video streaming models, while URLLC service takes the parameter settings of FTP 2 model~\cite{URLLC}, the same setting seen in~\cite{li2018deepReinforce}, similar settings also in~\cite{zhang2019macs}. The parameters are described in Table~\ref{table:paremeter} based on their respective streaming model.
 This BS is fixed and given a certain bandwidth (100 Mbps). It limits
the total available bandwidth.
System operators slice the network by assigning bandwidth to each type of user (a slice). 

%Users in the same slice are assigned bandwidth equally. 

Time slots are employed in the system (one second in each slot). At the start of each time slot,  the BS determines the amount of bandwidth $b_i$ will be allocated to $i$  ($i$ being the user type: Video, VoLTE, and URLLC) according to the number of active users in each slice.  By taking $t_i$ to represent the traffic demand for each slice, the throughput for each user type is $\min(b_i,t_i)$. CLARA aims to maximize the total throughput over time.
%The total throughput is the sum of throughput of all three network slices. 

% \newcommand{\xl}[1]{{\textcolor{red}{#1}}}
% \xl{change "unsatisfied ratio" to "dissatisfaction ratio".}

Each type of user is assigned a dissatisfaction ratio indicating how dissatisfied they are with their service. We set it as  $1-\frac{\min(b_i,t_i)}{t_i}$ in our simulation. The dissatisfaction ratio is 0 if the bandwidth assigned to the users exceeds or equals their demands. Otherwise, it will be $1-\frac{b_i}{t_i}$. The overall dissatisfaction ratio over time should under a constraint.
BS determines the latency $l_i$ of each user type using a queue maintained by the BS. Our system employs the first-come-first-serve principle. In queued traffic packets, those not processed during the current time slot are forwarded to the next period. Latency is the total time needed to transmit a traffic packet and it cannot be easily formulated mathematically. In each step, the latency should under a limitation.
 %Latency is considered as an instantaneous constraint. 
% The queue records remained traffic volume in each time slot for each type of users. In the case that the traffic volume can't be transmitted completely at its requested time slot. It will be moved to next time slot until it's complete. The queue also records how many time each type of user need to complete their traffic volume, which is latency. 
% The latency is a delay metric as we discussed in sec~\ref{sec:why}. 
%We perform first come first serve scheduling method. The uncompleted traffic volume will be severed firstly in the next time slot.

In addition,  dynamic hidden structures exist regarding how users engage in the corresponding service depending on the service quality. In general, users depart and arrive at the network according to Poisson distributions with mean of $\lambda_i$ and $\mu_i$. The arrival rate $\lambda_i$  is adaptive with the satisfaction ratio from the last time slot. In each time slot, $\lambda_i$  is updated with $\lambda_i = 0.99*\lambda_i + 0.01*\lambda_i*\frac{b_i}{t_i}$. 
Furthermore, users may depart early if the service quality is unsatisfactory. Specifically, we assume that if the dissatisfaction ratio is  0, no one  departs early. Otherwise, users depart with probability $1-\frac{b_i}{t_i}$.
\textbf{During the slicing process, we assume that these traffic demands and mobility patterns are unknown. This gives learning-based approaches, which incorporate exploration, an advantage over traditional methods which rely only on observed states. }

%There exist another hidden structure for early departure based on user experience: a user departs early with a certain probability depending on its dissatisfaction ratio. If the dissatisfaction ratio is less than 0, no one will departure. Otherwise, people departs will possibility $1-\frac{b\_i}{t\_i}$.

%% file: formulation.tex
\section{Problem Formulation}
\label{Sec:ProblemFormulation}
%Network slicing typically has three types of requirements: performance, functional, and operational. 
%From an operator’s view point, network slicing can be considered as a multi-objective optimization problem with constraints. 
In this section, we present the problem formulation for CLARA. 
Resource allocation in the network slicing is an optimization problem in essence. 
From the system operator's view point, the objectives include total throughput, Quality of Service (QoS), Quality of Experience (QoE), churn rate, as well as operational revenue and cost, etc. 
At the same time, network slicing is subject to a number of constraints, including resource constraints (e.g., spectrum, power, computation, storage), user service performance constraints (e.g., latency, average data rate,  QoS, QoE), security requirements (e.g., isolation and  firewall) and service level agreement (SLA) requirements (e.g., service availability and reliability), etc.
% In this work, we formulate the problem as a Constrained Markov Decision Process (CMDP). 

\begin{table*}[t]
	\centering
	\begin{tabular}{|c c c c|} 
		\hline
		Distribution& Initial number of users&Inter-Arrival Time&Packet Size\\
		 \hline
		 \multirow{2}{*}{Video}&\multirow{2}{*}{Poisson [Mean=50]}&Pareto [Exponential Para = 1.2, &Truncated Pareto [Exponential Para = 1.2, \\
		 &&Mean= 6 ms, Max = 12.5 ms]& Mean= 100 Byte, Max = 250Byte]\\
		 \hline
		 VoLTE& Poisson [Mean=50]  &Uniform [Min = 0, Max =160ms] &Constant (40 Byte) \\
		 \hline
		 \multirow{2}{*}{URLLC}&\multirow{2}{*}{Poisson [Mean=10]}&Truncated Exponential&Truncated Lognormal [Mean = 2 MB, \\
		 &&[Mean = 180ms]&Standard Deviation = 0.722 MB, Maximum =5 MB]\\
		\hline
	\end{tabular}
	\caption{Parameter for different types of user}
	\label{table:paremeter}
\end{table*}

\begin{table*}[t]
	\centering
	\begin{tabular}{|c c|} 
		\hline
		CMDP& Radio Resource Slicing \\
		\hline
		State& Number of active users in each type: $(n_{Video},n_{VoLTE},n_{URLLC})$  \\
		 \hline
	    Action&Sliced bandwidth to each type of user: $(b_{Video},b_{VoLTE},b_{URLLC})$\\
		 \hline
	    Reward function&Total throughput: $\min(b_{Video},t_{Video})+\min(b_{VoLTE},t_{VoLTE})+\min(b_{URLLC},t_{URLLC})$ \\
		\hline
		Explicit instantaneous constraint& Sum of sliced bandwidth: $b_{Video}+b_{VoLTE}+b_{URLLC}$\\
		\hline
		Implicit instantaneous constraint& Average of latency\\
		\hline
		Cost function&Dissatisfaction ratio: $1-\frac{b_i}{t_i}$, (i =Video, VoLTE  and URLLC)\\
		\hline
	\end{tabular}
	\caption{Mapping from Radio Resource Slicing to CMDP}
	\label{table:mapping}
% 	\vspace{-4mm}
\end{table*}
\subsection{Constrained  Markov Decision Process (CMDP)} \label{sec:cmdp}

%Mathematically, we formulate the problem as a constrained Markov Decision Process {CMDP} problem which can be efficiently solved by deep constrained reinforcement learning.

Mathematically, we formulate the problem as a CMDP, which % is an extension of the Markov Decision Process (MDP). The CMDP 
is represented with the tuple $(S,A,R, C,P,\mu, \gamma)$. 
The network observations constitute the state set $S$, e.g. current network slice allocation, network load (e.g., the number of users and traffic demand),  network status (e.g. cell conditions, neighboring cell interference level), etc.

The resource allocation decisions constitute the action set $A$, which depending on what level CLARA runs, could be admission control, computation and spectrum allocation, as well as network configurations, etc. 
 
The reward of taking action $a$ under state $s$ is defined as the reward function $R:S \times A \times S \mapsto \mathbb{R}$, e.g., throughput and revenue, etc. The objective is to maximize the cumulative reward expectation.
Similarly, the costs of taking action $a$ under state $s$ is defined as the cost functions, $C_i:S \times A \times S \mapsto \mathbb{R}$. There are $m$ cost functions and each is under a constraint, e.g. service level agreement (SLA) and latency, etc.
% $R:S \times A \times S \mapsto \mathbb{R}$ is the reward, which can be a weighted sum of the objectives. 
% including the overhead of network slice reconfiguration as a negative reward.
% There are $m$ cost functions with constraints; each function $C_i:S \times A \times S \mapsto \mathbb{R}$ cumulatively is under a constraint, such as service level agreement (SLA), outage probability and average data rate, etc.
 $P: S \times A \times S \mapsto [0,1]$ is the unknown transition probability function, where $P(s'|s,a)$ is the transition probability from state $s$ to $s'$ taking action $a$. 
 Last, $\mu: \mathit{S}\mapsto [0,1]$ is the initial state distribution and $\gamma$ is the discount factor, which can be different for reward and costs. 
 
 The actions are constrained by two types of constraints.
% , e.g., resource limits, fairness, and other network constraints. 
A cumulative constraint requires that the cumulatively sum of a cost is within a certain limit, e.g., outage probability or average throughput, etc, while an instantaneous constraint requires that a cost needs to satisfy a condition in each time slot, e.g. resource limits and service latency requirement, etc. 
% The instantaneous constraints can be further categorized into two types, explicit and implicit instantaneous constraints. 
Instantaneous constraints can be further divided into explicit and implicit instantaneous constraints.
%The instantaneous constraints include explicit instantaneous constraints and implicit instantaneous constraints. 
% The explicit instantaneous constraints can be accurately (and relatively easily) evaluated for each action, 
% for example, spectrum band available, the number of antennas, and transmission time, etc. 
An explicit constraint has a closed-form expression that can be numerically checked, e.g., transmission power and spectrum available, etc.
% The implicit instantaneous constraints are the outcome of actions that we do not have an accurate closed-form formulation. 
% Examples include latency (e.g., average latency or tail latency), job competition time, interference, etc. 
An implicit constraint does not have an accurate closed-form formulation due to the complexity of the system, e.g., latency and interference level, etc. 
Such constraints have to be modeled or learned using existing data and/or during exploration. 
We define action $a$ as feasible,
if $a\in A$ satisfies all the constraints including both cumulative and instantaneous. 

\subsection{Mapping Network Slicing to CMDP}
Mapping to  the radio access scenario, as summarized in Table~\ref{table:mapping}, the state $s=(n_{Video},n_{VoLTE},$
$n_{URLLC})$ is the  number of users observed at the beginning of each time slot. We do not know the exact traffic demand generated in this time slot.
The action $a=(b_{Video},b_{VoLTE},b_{URLLC})$ is the bandwidth allocation for each type of users.
The reward $R(s,a)$ at each time slot is the total throughput $\min(b_{Video},t_{Video})+\min(b_{VoLTE},t_{VoLTE})+\min(b_{URLLC},t_{URLLC})$. 
% The objective is to maximize the long-term reward expectation. % as in Eq.~(\ref{objective}).

Moreover, each type of users has a cumulative constraint, which is  the expectation of cumulative dissatisfaction ratio. The corresponding cost function $C_i(s,a)$ in each step is the dissatisfaction ratio $1-\frac{\min(b_i,t_i)}{t_i}$ for each type of user.  

Last, the actions need to satisfy  both the explicit and implicit instantaneous constraints. The explicit instantaneous constraint is the sum of allocated bandwidth, $b_{Video}+b_{VoLTE}+b_{URLLC}$, which must be less or equal to the total bandwidth ($100$ Mbps). The implicit instantaneous constraint is on the average latency of each type of user, which we cannot get a closed-form solution and needs to be learned. The average latency should be less or equal to a predefined limitation.

CLARA learns a policy $\pi$ takes states as input and output actions. We denote a policy as $\pi_{\theta}$ to emphasize its dependence on the parameter $\theta$ (e.g., a neural network policy) and $\tau = (s_{0}, a_{0},s_{1}, a_{1}... )$ is a trajectory, where $\tau \sim \pi_{\theta}$.
The objective is to select a policy $\pi_{\theta}$, which maximizes the discounted cumulative reward
\begin{equation}\label{Def:Reward}
    J_{R}^{\pi_{\theta}}= \mathbb{E}_{\tau \sim \pi_{\theta}}[\sum_{t=0}^{\infty}\mathit{\gamma}^{t}\mathit{R}(s_{t},a_{t},s_{t+1})],
\end{equation}
while satisfying discounted cumulative constraints
\begin{equation}
    J_{C_{i}}^{\pi_{\theta}}=\mathbb{E}_{\tau \sim \pi_{\theta}}[\sum_{t=0}^{\infty}\gamma^{t}\mathit{C_i}(s_{t},a_{t},s_{t+1})]
\end{equation}
and instantaneous constraints.
% , defined as 
% \begin{equation}
%   J_{R}^{\pi_{\theta}} = \mathbb{E}_{\tau \sim \pi_{\theta}}[\sum_{t=0}^{\infty}\mathit{\gamma}^{t}\mathit{R}(s_{t},a_{t},s_{t+1})],
% \end{equation}
% \begin{equation}
%     J_{C_{i}}^{\pi_{\theta}} = \mathbb{E}_{\tau \sim \pi_{\theta}}[\sum_{t=0}^{\infty}\gamma^{t}\mathit{C_i}(s_{t},a_{t},s_{t+1})], \text{for each  }C_i, 
% \end{equation}
% where $\tau = (s_{0}, a_{0},s_{1}, a_{1}... )$ is a trajectory, and $\tau \sim \pi_{\theta}$. 

Formally, the optimization problem is defined as  
\begin{align}
& \underset{\theta}{\text{maximize}} & &\max_{\theta} J_{R}^{\pi_{\theta}}\label{eq:objective}  \\
&\text{subject to}  & & J_{C_{i}}^{\pi_{\theta}} \leq \omega_i,  \text{for each i}, \label{eq:constraint}  \\
& & & C_{j}(s_t, a_t)\leq \epsilon_j, \text{for each j and t}, \label{eq:instconstraint} 
\end{align}

The formulation and following CLARA algorithm can be applied to general network slicing problems, including radio access, edge, cloud, and core networks. 
Here, we explain CLARA on a radio resource slicing scenario derived from the work~\cite{li2018deepReinforce} and it can be extended to more general cases easily.

%% file: preliminary.tex
\section{Preliminaries}
%\subsection{Reinforcement learning}
%Reinforcement learning learns how a agent takes actions to achieve maximum expected cumulative return by interacting with the environment. Typically, it doesn't consider constraints. It has been successfully applied in solving complex sequential decision and control problems. For example, robot control, self driving, and gaming, etc. The reinforcement learning problem can be formalized as a Markov Decision Process (MDP). The basic three elements are state, action and reward function. State is the observation from environment. Action is the response of agent based on the observation. And reward function is the feedback from the environment telling if the action is good or not.
%
%Deep reinforcement learning estimates the value of expected cumulative return with neural network instead of traditional linear function approximation. Because the linear function approximation can not accurately model the value any more with the increasing complex of the problem. To solve the complex problem, there exist two types of methods: valued based reinforcement learning and policy based reinforcement learning. Valued based methods aims to estimate the value for each state action pair and update the policy iterative to maximum the long term return. Policy-based methods start from the view of optimization. It update the policy with the gradient of the long term return over policy estimation parameter. In our work, we are based on the policy gradient methods.
%
To solve the above problem, we briefly review the preliminaries from previous work in this section.
% In this section, we briefly review the preliminaries from previous work to solve our problem. 
\subsection{Definition}
%We first go through definitions in CMDP.  
For a state-action trajectory starting from state $s$,  the value function of state $s$ is
\begin{equation*}
    V_{R}^{\pi_{\theta}}{(s)} = \mathbb{E}_{\tau \sim \pi_{\theta}}[\sum_{t=0}^{\infty}\gamma^{t}\mathit{R}(s_{t},a_{t},s_{t+1})|s_{0}=s]. 
\end{equation*}
The  action-value function of state $s$ and action $a$ is 
\begin{equation*}
    Q_{R}^{\pi_{\theta}}{(s,a)} = \mathbb{E}_{\tau \sim \pi_{\theta}}[\sum_{t=0}^{\infty}\gamma^{t}\mathit{R}(s_{t},a_{t},s_{t+1})|s_{0}=s, a_{0}=a],
\end{equation*}
and the advantage function is 
\begin{equation}\label{advantage}
    A_{R}^{\pi_{\theta}}{(s,a)} = Q_{R}^{\pi_{\theta}}{(s,a)} - V_{R}^{\pi_{\theta}}{(s)}. 
\end{equation}
In the CMDP, the corresponding values for the cumulative constraint cost functions are calculated in the same way.

$V_{C_i}^{\pi_{\theta}}{(s)}$, $Q_{C_i}^{\pi_{\theta}}{(s,a)}$, $A_{C_i}^{\pi_{\theta}}{(s,a)}$, 
for each constraint cost function $C_i$,  
are defined by 
%analogy to $V_{R}^{\pi_{\theta}}{(s)}, Q_{R}^{\pi_{\theta}}{(s,a)}, A_{R}^{\pi_{\theta}}{(s,a)}$, 
replacing reward function $R$ above with $C_i$. 
In the following sections, to simplify the notation, we omit the subscripts of $R$ and $C_i$ if there is no ambiguity. 

Let $\rho_{\pi_\theta}(s)$ be the discounted visitation frequencies 
\begin{equation*}
    \rho_{\pi_\theta}(s) = \sum_{t=0}^\infty \gamma^t P(S_t=s),  
\end{equation*}
where the actions are chosen according to $\pi_\theta$. 

The Kullback-Leibler (KL) divergence~\cite{kullback1951information} of  distribution $\Gamma$ from $\Delta$ is defined as
\begin{equation*}
    D_{KL}(\Gamma, \Delta) = \sum_{x\in\chi} \Gamma(x)\log \left( \frac{\Gamma(x)}{\Delta(x)}\right). 
\end{equation*}
We denote $D_{KL}^{max}(\pi, \pi') = \max_{s} D_{KL}(\pi(\cdot|s),\pi'(\cdot|s)) $. 

\subsection{Policy Gradient Methods}
% Policy gradient methods target at modeling the reinforcement problem as a optimization problem and solving it directly by updating the policy with the gradient of policy estimation parameter.
%

 Policy gradient~\cite{sutton2000policy} method is applied to find an optimal policy of an unconstrained Markov Decision Process (MDP) problem. It  calculates the gradient of the objective Eq. (\ref{Def:Reward}),  
 \begin{equation*}
     \bigtriangledown J^{\pi_{\theta}} = \mathbb{E}_{t}[\bigtriangledown_{\theta} log\pi_{\theta}(a_{t}|s_{t})A_{t}]
 \end{equation*}
 where $\pi_{\theta}$ is the current policy under parameter $\theta$ and $A_{t}$ is the advantage function Eq. (\ref{advantage}) at time step $t$. Thereafter, $\theta$ is updated as 
 \begin{equation*}
     \theta = \theta + \eta \bigtriangledown J^{\pi_{\theta}},
 \end{equation*}
 where $\eta$ is the learning rate. 
 
 Trust Region Policy Optimization (TRPO)~\cite{schulman2015trust} is proposed to achieve monotonic improvement of the new policy based on the results of the previous policy. 
 The objective is approximated with a surrogate function combined with the Kullback Leibler (KL) divergence shown as follows. 
 We denote a local approximation $L_{\pi_{\theta}}(\pi_{\theta'})$ for $J^{\pi_{\theta'}}$ with $\pi_{\theta}$ as 
 \begin{equation}
    L_{\pi_{\theta}}(\pi_{\theta'}) =  J^{\pi_{\theta}} + \sum_s \rho_{\pi_\theta}(s) \sum_a \pi_{\theta'} (a |s) A^{\pi_{\theta}}(s,a). 
 \end{equation}
 The objective of TRPO is to maximize 
 \begin{equation}\label{trpo}
     \centering
     \begin{split}
      %\max_{\theta}~
      L^{TRPO}(\theta) =  L_{\pi_{\theta_{o}}}(\pi_{\theta}) - \frac{4\varepsilon^{\pi_{\theta_{o}}}\gamma}{(1-\gamma)^2} D_{KL}^{max}(\pi_{\theta_{o}}, \pi_{\theta}),  %\mathbb{E}_{t}[\frac{\pi_{\theta}(a_{t}|s_{t})}{\pi_{\theta_{o}}(a_{t}|s_{t})}A_{t}]\\
    % s.t.~~~\mathbb{E}_{t}[KL[\pi_{\theta_{o}}(a_{t}|s_{t}), \pi_{\theta}(a_{t}|s_{t})]] \leq \delta .
 \end{split}
 \end{equation}
 where $\pi_{\theta_{o}}$ is the old policy to improve,  $\varepsilon^{\pi_{\theta_{o}}}=\max_{s,a} |A^{\pi_{\theta_{o}}}(s,a)|$. 
 %$\delta$ is the step size limitation. 
 The reward improvement of the new policy obtained by solving the optimization is guaranteed by the following theorem. 
 \begin{theorem}\label{Thm:TRPO}~\cite{schulman2015trust}
 $$J^{\pi_{\theta}} \geq L_{\pi_{\theta_{o}}}(\pi_{\theta}) - \frac{4\varepsilon^{\pi_{\theta_{o}}}\gamma}{(1-\gamma)^2} D_{KL}^{\max}(\pi_{\theta_{o}}, \pi_{\theta}). $$ 

 \end{theorem}
 %TRPO takes advantage of a linear approximation to the objective and a quadratic approximation to the KL divergence constraint, 
 %to solve the reduced approximation problem with conjugate gradient optimization efficiently. 
 
 The following theorem provides a tighter bound than Theorem~\ref{Thm:TRPO} to measure the difference between any two policies.  
 \begin{theorem}\label{thm:CPO}~\cite{achiam2017constrained} 
 For any policies $\pi', \pi$, the following bound holds:
 \begin{align*}
     &J(\pi') - J(\pi) \geq \\
      & \mathop{\mathop{\mathbb{E}}_{s\sim\rho_{\pi_{\theta}}}}_{ a\sim \pi_{\theta'}} \left[ A^\pi(s,a)\right]
     -\frac{\gamma\varepsilon^{\pi'}}{(1-\gamma)^2}\sqrt{2D_{KL}^{\max}(\pi', \pi)},
 \end{align*}
 where
 $\varepsilon^{\pi'}=\max_{s} \left |\mathbb{E}_{a\sim \pi'} [A^{\pi}(s,a)] \right|.$
 
 \end{theorem}
 
  Proximal Policy Optimization (PPO)~\cite{schulman2017proximal} is a heuristic algorithm with the same intuition as TRPO to formulate the problem with a first-order surrogate optimization to reduce the complexity, maximizing 
 \begin{equation}
   \begin{aligned} \label{PPO}
   %&\max_{\theta}~
   &L^{CLIP}(\theta ) =  \\
   &~\mathop{\mathop{\mathbb{E}}_{s\sim\rho_{\pi_{\theta_{o}}}}}_{ a\sim \pi_{\theta_{o}}}[\min\left\{r(\theta),
 \operatorname{clip}(r(\theta), 1-\varepsilon, 1+\varepsilon)\right\}A^{\pi_{\theta_{o}}}(s,a)], 
 \end{aligned} 
 \end{equation}
 where $r(\theta) = \frac{\pi_{\theta}(a|s)}{\pi_{\theta_{o}}(a|s)}$, $\operatorname{clip}(\cdot)$ is the clip function and $r_{t}(\theta)$ is clipped between $\left [ 1-\epsilon, 1+\epsilon  \right ]$.

%% file: method.tex
\section{Constrained Reinforcement Learning}

\subsection{Cumulative Constraints}
%\begin{figure}[t]
%	\centering
%	\includegraphics[width=0.81\columnwidth]{figures/indicator.pdf}
%	\caption{Value of indicator function $I(x)$ and logarithmic barrier functions $\phi(x) = \frac{log(-x)}{t}$ with $t=20$ and $t=50$.
%	}
%	\label{fig:indicator}
%\end{figure}

To deal with the cumulative constraints in the CMDP, we develop our method built upon our previous work Interior-point Policy Optimization
(IPO)~\cite{liu2020ipo}. IPO 
 augments the objective function of PPO (Eq. (\ref{PPO})) with logarithmic barrier functions without theoretical guarantee for the policy improvement. However, in this work, 
 we incorporate the TRPO objective with logarithmic barrier functions and provide  theoretical policy improvement guarantees. 
 Moreover, IPO uses a fixed hyperparameter $t$ for the logarithmic barrier function, while we further propose practical algorithms in an adaptive manner. 

%IPO is inspired by the interior-point method~\cite{boyd2004convex}. 
The way that we deal with constraints is inspired by the interior-point method~\cite{boyd2004convex}. 
For completeness, we briefly review the related background. 
Intuitively, the constrained optimization problem is reduced to an unconstrained one by adding indicator functions  as
penalty 
such that 1) if a constraint is satisfied, zero penalty is added to the objective function, and 2) if the constraint is violated, the penalty is negative infinity, ideally. 
%\begin{equation*}
%    I{(\widehat{J}_{C_{i}}^{\pi_{\theta}})} = \begin{cases}
% & 0\ \ \ \ \  \ \widehat{J}_{C_{i}}^{\pi_{\theta}}\leq 0,\\ 
% & -\infty\ \ \widehat{J}_{C_{i}}^{\pi_{\theta}}>  0. 
%\end{cases}
%\end{equation*}

A logarithmic barrier function is a differentiable approximation of the indicator function, defined as 
\begin{equation*}\label{log}
    \phi {(\widehat{J}_{C_{i}}^{\pi_{\theta}})} = \frac{\log(-\widehat{J}_{C_{i}}^{\pi_{\theta}})}{t}, 
\end{equation*}
where $t$ is a hyperparameter, 
and 
 $\widehat{J}_{C_{i}}^{\pi_{\theta}} = J_{C_{i}}^{\pi_{\theta}} - \omega_{i}$, to simplify the notation. 
% It's a penalty functions to accommodate the constraints. is easy to implement, and also provide nice analytical properties; thus adopted. 
%As shown in Fig. \ref{fig:indicator}, 
We get a better approximation for the indicator function with a higher $t$.

Now we take $L^{TRPO}(\theta)$ in Eq. (\ref{trpo}) as our objective with cumulative constraints, that is, 
\begin{equation}\label{PPO_constraints}
    \begin{split}
    \max_{\theta}~L^{TRPO}(\theta ), \\
    s.t.~~~\widehat{J}_{C_{i}}^{\pi_{\theta}} \leq 0 .
\end{split}
\end{equation}
% We assume that the above optimization problem is strictly feasible and reduce it  
We reduce the above optimization problem to an unconstrained optimization by augmenting the objective with the logarithmic barrier functions for constraints. Our objective becomes 
\begin{equation} \label{TRPO-IPO}
    \max_{\theta}~L^{TRPO}(\theta )+\sum_{i=1}^{m}  \phi {(\widehat{J}_{C_{i}}^{\pi_{\theta}})}.  
\end{equation} 
In the following subsection, we show the theoretical guarantee of the policy performance with this objective function. 
%In other words, our objective function becomes
\subsection{Policy Performance Bounds} 
\label{Sec:policyperformance}
\begin{theorem}
\label{thm:IPO}
  The maximum gap between the optimal value of Eq. (\ref{PPO_constraints})  and the optimal of Eq. 
  (\ref{TRPO-IPO}) is bounded by $\frac{m}{t}$, where $m$ is the number of constraints and $t$ is the hyperparameter of logarithmic barrier function.
\end{theorem}

This theorem provides a lower bound for the performance guarantee for incorporating the logarithmic barrier function into the TRPO objective. 
The proof is similar as the proof of the guarantee for the PPO objective in Theorem 1 of~\cite{liu2020ipo}.  
The theorem shed lights upon the effects from the hyperparameter $t$, since a larger $t$ provides a better approximation of the original objective, which is also validated in the empirical experiments in  IPO~\cite{liu2020ipo}. % also show that we can get a larger long-term reward with a larger $t$.
In practice, a larger $t$ can result in a large fluctuation near the boundary. Hence, there is a tradeoff in choosing this hyperparameter $t$.

\begin{theorem}
\label{thm:PolicyImprove}
Assume $\pi_{\theta}$ is obtained by solving Eq. (\ref{TRPO-IPO}) based on a previous feasible policy $\pi_{\theta_{o}}$. The following bounds holds 
 \begin{tiny}
\begin{align*}
       & J_R^{\pi_{\theta}} - J_R^{\pi_{\theta_{o}}} \geq \\
        &  -\sum_{i\in \Phi} \frac{1}{t} \log\left( 1 - \mathop{\mathop{\mathbb{E}}_{s\sim\rho_{\pi_{\theta_{o}}}}}_{ a\sim \pi_{\theta}} \left[\frac{1}{\psi^{\pi_{\theta_o}}} A_{C_i}^{\pi_{\theta_{o}}}(s,a)\right]
     +\frac{\sqrt{2}\gamma\varepsilon^{\pi_\theta}_{C_i}}{{\psi^{\pi_{\theta_o}}}(1-\gamma)^2}\sqrt{D_{KL}^{\max}(\pi_{\theta}, \pi_{\theta_{o}})}\right ),
\end{align*}
\end{tiny}

%\mathop{\mathop{\mathbb{E}}_{s\sim\rho_{\pi_{\theta_{o}}}}}_{ a\sim \pi_{\theta}} \left[ A_{C_i}^{\pi_{\theta_{o}}}(s,a)\right]
where
$\Phi=\{i:J_{C_i}(\pi_\theta)<J_{C_i}(\pi_{\theta_{o}})\}$, 
$$\varepsilon^{\pi_\theta}_{C_i}=\max_{s} \left |\mathbb{E}_{a\sim \pi_\theta} [A^{\pi_{\theta_{o}}}_{C_i}(s,a)] \right|,$$
 $$\psi^{\pi_{\theta_o}}=\min_i (\omega_i-J_{C_i}^{\pi_{\theta_o}}).  $$
 \end{theorem}

%\pi_{\theta_{o}}
%\pi_\theta
\begin{proof}
Since $\pi_\theta$ is obtained by maximizing the objective in Eq. (\ref{TRPO-IPO}) and $\pi_{\theta_{o}}$ is a feasible policy, we have
\begin{small}
\begin{align} 
    &J_R(\pi_{\theta_{o}}) + \frac{1}{t} \sum_i \log\left(\omega_i-J_{c_i}(\pi_{\theta_{o}}) \right)      \nonumber  \\
    =& L_{\pi_{\theta_{o}}}(\pi_{\theta_{o}}) - \frac{4\varepsilon^{\pi_{\theta_{o}}}\gamma}{(1-\gamma)^2} 
    D_{KL}^{max}(\pi_{\theta_{o}}, \pi_{\theta_{o}}) + \frac{1}{t} \sum_i \log\left(\omega_i-J_{c_i}(\pi_{\theta_o}) \right)      \nonumber   \\
%    =& L^{IPO}(\theta_{o}) \leq  L^{IPO}(\theta) \nonumber \\
    \leq &L_{\pi_{\theta_{o}}}(\pi_{\theta}) - \frac{4\varepsilon^{\pi_{\theta_{o}}}\gamma}{(1-\gamma)^2} 
    D_{KL}^{max}(\pi_{\theta_{o}}, \pi_{\theta}) + \frac{1}{t} \sum_i \log\left(\omega_i-J_{c_i}(\pi_\theta) \right) \nonumber \\
    \leq  &J_R(\pi_\theta) + \frac{1}{t} \sum_i \log\left(\omega_i-J_{c_i}(\pi_\theta) \right). \label{step:TRPO}
\end{align}
\end{small} 
Theorem \ref{Thm:TRPO} is applied in Step (\ref{step:TRPO}). 

\begin{align}
   & J_R^{\pi_{\theta}} - J_R^{\pi_{\theta_{o}}} \nonumber \\
   \geq & -\sum_i \frac{1}{t} \log \left( 1+\frac{ J_{C_i}^{\pi_{\theta}} - J_{C_i}^{\pi_{\theta_{o}}}}{\omega_i-J_{C_i}^{\pi_{\theta_{o}}}}  \right) \nonumber \\
   \geq & -\sum_{i\in \Phi} \frac{1}{t} \log \left( 1+\frac{ J_{C_i}^{\pi_{\theta}} - J_{C_i}^{\pi_{\theta_{o}}}}{\omega_i-J_{C_i}^{\pi_{\theta_{o}}}} \right)  \label{throwPos} \\
   \geq &  -\sum_{i\in \Phi} \frac{1}{t} \log\biggl( 1 - \mathop{\mathop{\mathbb{E}}_{s\sim\rho_{\pi_{\theta_{o}}}}}_{ a\sim \pi_{\theta}} \left[\frac{1}{\psi^{\pi_{\theta_o}}} A_{C_i}^{\pi_{\theta_{o}}}(s,a)\right] \nonumber\\
     &+\frac{\sqrt{2}\gamma\varepsilon^{\pi_\theta}_{C_i}}{{\psi^{\pi_{\theta_o}}}(1-\gamma)^2}\sqrt{D_{KL}^{\max}(\pi_{\theta}, \pi_{\theta_{o}})} \biggr), \label{step:CPO}
\end{align}
In Step (\ref{throwPos}), we neglect all the positive terms with  $J_{C_i}(\pi_\theta)\geq J_{C_i}(\pi_{\theta_{o}})$. 
Step (\ref{step:CPO}) is obtained by applying Theorem \ref{thm:CPO}. 
\end{proof}
In the above theorem, we provide new results for the lower bound of our policy improvement. 
The worst case performance guarantee is decided by the advantage functions of the previous policy and the KL divergence between it and the new policy.  A lower the KL divergence and higher advantage function values can lead to a better policy improvement. 

\subsection{Adaptive IPO}
Besides the theoretical improvements, we propose our adaptive IPO algorithm, improved over the original IPO in~\cite{liu2020ipo}. 
As demonstrated in Theorem \ref{thm:IPO}, there is a trade-off between approximation accuracy and algorithm performance depending on the hyperparameter $t$. 
In the original IPO, the hyperparameter $t$ is fixed. It is time-consuming to tune the hyperparameter and is not adaptive to different system dynamics. In our work, we adjust the hyperparameter in an adaptive manner: we start with a small $t$ to have  more stable policy updates, and gradually increase $t$ to achieve better policies on convergence.  

In practice, we face the challenge of 
% slow convergence 
complex computation
of TRPO in large-scale problems. 
The heuristic algorithm, PPO, achieves better performance than TRPO empirically, which is employed in IPO.  
In the implementation, we also replace the objective function of TRPO with that of PPO in Eq. (\ref{PPO}), making the objective as maximizing 
\begin{equation} \label{IPO}
    L^{IPO}(\theta)=L^{CLIP}_R(\theta )+\sum_{i=1}^{m}  \phi {(\widehat{J}_{C_{i}}^{\pi_{\theta}})},
\end{equation}
where we denote $L^{CLIP}_R(\theta )$ and $L^{CLIP}_{C_i}(\theta )$  to be the formulations applying Eq. (\ref{PPO}) to reward and cost functions. 
%With the logarithmic barrier functions, all constraints can be properly incorporated into the reward. 
% Furthermore, for better constraint satisfaction during the training process,  we build upon PPO~\cite{schulman2017proximal}, which inherits the trust region property of TRPO to guarantee constraint satisfaction during training. 
%
%Furthermore, to take advantage the new objective, our policy should starts from a feasible point which satisfies the cumulative constraints. Typically, the actor network start from randomly initialization. To address this problem, before optimizing the IPO (Eq. \ref{IPO}), we have a phase called phase I to find a feasible initial policy. In phase I, we regard the cumulative constraints as our objective and build upon PPO (Eq. \ref{PPO}) to decrease the constraint until we find a feasible solution. 

We present our adaptive IPO in Algorithm \ref{alg:IPO}. 
In Phase I, the cost functions are successively optimized to obtain a feasible policy.  
The algorithm is initialized with the objective of maximizing 
\begin{equation*}
    L^C(\theta) = -L_{C_1}^{CLIP}(\theta),
\end{equation*}
to decrease the cumulative cost $C_1$ until the constraint is satisfied. 
%In our algorithm, we use the stochastic gradient descent to obtain the optimal parameters. 
Thereafter, at the end of iteration $i$, we update the objective $L^C(\theta)$ for next iteration, 
$$L^{C}(\theta) = -L^{CLIP}_{C_{i+1}}(\theta) + \sum_{j=1}^{i} \phi {(\widehat{J}_{C_{j}}^{\pi_{\theta}})},$$ 
where 
we replace $-L_{C_i}^{CLIP}(\theta)$ with its logarithmic barrier function
$\phi {(\widehat{J}_{C_{i}}^{\pi_{\theta}})}$ to guarantee the constraint for $C_i$ is satisfied in the future updates, 
and add $-L_{C_{i+1}}^{CLIP}(\theta)$ for the next cost function $C_{i+1}$. 
The above process is repeated until obtaining a feasible policy for all the constraints. 

In Phase II, the algorithm is initialized with the feasible policy in Phase I. 
%As we discussed in Section \ref{Sec:policyperformance}, a larger $t$ results in better accuracy of the policy estimation and meanwhile greater fluctuation near the boundary of the feasible region. 
In light of the trade-off of $t$, 
we start with a moderate small $t$ and adaptively increase it with a factor $\mu>1$ when policy convergence criteria are satisfied. 
In each iteration, we update  the policy by maximizing the objective  $L^{IPO}(\theta)$ in Eq. (\ref{IPO}).  

% The adaptive IPO acts as the main procedure to train our policy neural network in CLARA.  

%In practice, the policy $\pi_{\theta}$ is approximated by a neural network with parameter $\theta$. %, whose input is the state and output is the action. 

%The objective for phase I is
%
% \begin{equation}
%   \begin{aligned} \label{phaseI}
%   \min_{\theta}~&L_C^{CLIP}(\theta )  \\
%  = &~\mathbb{E}_{t}[\min(r_{t}(\theta)A_{t}^{C},
% \operatorname{clip}(r_{t}(\theta),1-\epsilon ,1+\epsilon)A_{t}^{C})], 
% \end{aligned} 
% \end{equation}
% where $r_t(\theta) = \frac{\pi_{\theta}(a_{t}|s_{t})}{\pi_{\theta_{o}}(a_{t}|s_{t})}$, $A_{t}^{C}$ is the advantage function for cost, $\operatorname{clip}(\cdot)$ is the clip function and $r_{t}(\theta)$ is clipped between $\left [ 1-\epsilon, 1+\epsilon  \right ]$. The optimization will stop when $\widehat{J}_{C_{i}}^{\pi_{\theta}} \leq 0$ which means the policy is in feasible space considering cumulative constraints. 
%

%A greater hyperparameter $t$ provides a better approximation of the original objective. However, it leads to a high gradient fluctuation near the boundary of the feasible region, which makes the policy harder to satisfy the constraint. So in the algorithm, we make the log-barrier extension gradually harder by increasing parameter t. The pseudo-code our method is shown in Algorithm \ref{alg:IPO}.

\begin{algorithm}[t]
	\caption{The procedure of adaptive {IPO}} \label{alg:IPO}
	\begin{flushleft}
	\hspace*{0.02in}{\bf Input:} 
	Initialize with a random policy $\pi_{\theta_{0}}$. 
	Set the hyperparameter  $\varepsilon$ for PPO, 
	$t=t_0$, $\mu>1$ for logarithmic barrier function, 
	and iteration number k = 0 
	
	\hspace*{0.02in}{\bf Output:} 
	The policy parameter $\theta$
	\end{flushleft}
	{\bf Phase I:}
	
	\begin{algorithmic}[1]
		\STATE Initialize $L^{C}(\theta) = -L^{CLIP}_{C_1}(\theta)$% the computational graph structure. 
		\FOR {i=1,2,\dots, m}
		 
		\WHILE{$\widehat{J}_{C_{i}}^{\pi_{\theta}} > 0$}
		\STATE Sample N trajectories $\tau_{1}, ..., \tau_{N}$ including observations, actions and  costs with the current policy $\pi_{\theta_{k}}$ 
		\STATE Calculate advantages, constraint values, etc
		\STATE Update $\theta_{k+1}$ by maximizing $L^{C}(\theta)$ with stochastic gradient descent methods 
		%$\theta_{k} +  \alpha  \bigtriangledown_{\theta} L^C (\theta ) $, where $\alpha$ is learning rate 
		\STATE Iteration k = k+1 
		\ENDWHILE
		\STATE Update the objective as maximizing $L^{C}(\theta) = -L^{CLIP}_{C_{i+1}}(\theta) + \sum_{j=1}^{i} \phi {(\widehat{J}_{C_{j}}^{\pi_{\theta}})}$
		\ENDFOR
	\end{algorithmic}
	
	{\bf Phase II:}
	
	\begin{algorithmic}[1]
		\FOR {iteration k}
		\STATE Sample N trajectories $\tau_{1}, ..., \tau_{N}$ including observations, actions, rewards and  costs with the current policy $\pi_{\theta_{k}}$ 
		\STATE Calculate advantages, constraint values, etc
		\STATE  Update $\theta_{k+1}$ by maximizing $L^{IPO}(\theta)$ with stochastic gradient descent methods 
		%= \theta_{k} +  \beta \bigtriangledown_{\theta} L^{IPO}(\theta) $, where $\beta$ is learning rate
		\IF {the policy converges}
		\STATE $t=\mu*t$ 
		\ENDIF
		\STATE Iteration k = k+1
		\ENDFOR 
		\RETURN the policy parameter $\theta=\theta_{k+1}$
	\end{algorithmic}
\end{algorithm}

\subsection{Instantaneous Constraints} 
We further extend the policy neural network with new layers to handle instantaneous constraints. Our policy network $\pi_\theta$ takes a state $s$ as input and outputs an action $a=\pi_\theta(s)$.
 %We can also handle both explicit and implicit instantaneous constraints based on previous works \cite{dalal2018safe,bhatia2019resource}. 
% \begin{figure}[t]
% 	\centering
% 	\includegraphics[width=0.9\columnwidth]{figures/safelayer.png}
% 	\caption{Action projection model}
% 	\label{fig:safelayer}
% \end{figure}
% \begin{figure}[t]
% 	\centering
% 	\includegraphics[width=0.9\columnwidth]{figures/linear.png}
% 	\caption{Model of estimate implicit instantaneous constraint}
% 	\label{fig:linear}
% \end{figure}
% \paragraph{Implicit instantaneous constraints.} 
To satisfy instantaneous constraints, one way is to project the infeasible action generated by $\pi_\theta$ to the feasible space~\cite{dalal2018safe}. To achieve this, one can introduce another additional last layer to $\pi_\theta$, whose role is to solve
% Denote the policy by $\pi_\theta$. Then, on top of the actor network we compose an additional, last layer, whose role is to solve
\begin{equation}
   \begin{aligned} \label{eq:implicit}
  & \min_{a} \frac{1}{2} \left \| a - \pi_\theta(s) \right \|^{2}  \\
  & s.t.~C_j(s_t,a_t) \leq \epsilon_j
 \end{aligned} 
 \end{equation}
where $C_j(s_t,a_t)$ is the instantaneous cost under new action, and $\epsilon_j$ is the predefined constraint. In other words, we project the action from the policy $\pi_\theta(s)$ to the $\ell_2$ nearest feasible action $a$ that satisfies the instantaneous constraint. The projection idea can apply to both explicit and implicit constraints.
% This layer, which we refer to as safety layer, perturbs the original action as little as possible in the Euclidean norm in order to satisfy the implicit instantaneous constraints.
One challenge for implicit constraints is that the function $ C_{j}(\cdot,\cdot)$ is unknown. To address the problem, we take advantage of another neural network to learn the value of $ C_{j}(s_t,a_t)$ simultaneously, as in~\cite{dalal2018safe}.

% a delay quality, which is latency in our scenario, we can't know the constraint value immediately. It's will lead us fail to solve the Eq. (\ref{eq:implict}). To address the problem, we take advantage another neural network to learn the value of $ c_{i}(s,a)$ (latency). 
% Inspired by \cite{dalal2018safe}, we choose a more elegant approach that comes with significant advantages to solve Eq. (\ref{eq:implict}), as shown in Fig. \ref{fig:linear}. Instead of learning the immediate-constraint functions $ c_{i}(s,a)$, we learn the advantage between two continuous sates, which can simplify the process of solving Eq. (\ref{eq:implict}) proved at \cite{dalal2018safe}.

\textbf{Global Sum constraints} %As stated in formulation~\ref{sec:cmdp}, we take a neural network with parameter $\theta$ to approximate the policy, which is denoted as $\pi_\theta$.  
% Our policy $\pi_\theta$ takes a state $s\in S$ as input and outputs an action $a=\pi_\theta(s)\in A$. 
 %Typically, the final layer in the policy network uses a tanh activation function that is scaled to give values from 0 to 1. 
%We handle the explicit instantaneous constraints, which needs the sum of actions equal to a limit, by having one additional layer at the end of the policy network, the traditional softmax layer~\cite{bhatia2019resource}. 
 are a special case of explicit instantaneous constraints that can be easily satisfied with a softmax projection layer.
Assume that the action is a vector with $k$ entries, denoted as $a=[a_1,a_2,...,a_k]$, and we deal with the specific explicit constraints that require $\sum_k a_k =1$. 
A softmax layer is added right after the output layer of policy network $\pi_\theta$ to project the arbitrary action $a$ to a feasible action $a'=[a'_1,a'_2,...,a'_k]$, where
\begin{equation}
    a'_i = \frac{e^{a_i}}{\sum_{j=1}^{k}e^{a_j}}. 
\end{equation}
In the radio access scenario~\ref{sec:scenario}, we handle our explicit constraints (bandwidth allocation) with the softmax projection layer and handle the implicit constraints (latency) with general $\ell_2$ projection layer.

%% file: experiment.tex
\section{Experiments}
\label{sec:experiments}
In the experiments, we demonstrate that CLARA outperforms the baselines including traditional methods and RL-based methods; CLARA can handle multiple constraints simultaneously and 
% , with higher long-term reward, as well as satisfying all the constraints; 
CLARA is able to speed up converge by incorporating domain knowledge;
at last, we discuss the reason why CLARA works well 
    through our observations. 
%\begin{itemize}
%    \item CLARA outperforms the baselines including traditional methods and RL-based methods, with higher long term reward, 
%    as well as satisfying all the constraints.  
%    %can get higher long term reward as well as satisfying both cumulative and instantaneous constraints, compared with traditional methods and unconstrained RL-based methods.
%    \item CLARA outperforms the baselines on constraint satisfaction.  
%    \item CLARA can incorporate domain knowledge to speed up converge.
%    %\item We discuss why our algorithm works with observations.
%\end{itemize}

\subsection{Settings}
%To demonstrate the advantage of constrained RL. 
% As described in Section~\ref{Sec:ProblemFormulation},  we allocate the bandwidth resource of a base station to three types of users:  Video, VoLTE and URLLC. 
% In this problem, the reward is the total throughput from the BS to users; 
% the cumulative constraints are the cumulative dissatisfaction ratio; 
% the implicit instantaneous constraints are service latency requirements. 
% % The system architecture  is described in Section~\ref{Sec:sys}. 
% The input states  of CLARA is the number of users of each type observed at the beginning of each time slot; the action is the bandwidth allocation to the three types. 

We evaluate the performance of CLARA on the radio resource slicing scenario, as described in Section~\ref{sec:scenario},
% We compare the performance of CLARA to  traditional benchmarks and RL-based methods.
and compare it with traditional benchmarks and RL-based methods.
The network scenario is numerically simulated.
As hidden structures and system dynamics exist, accurate models of the system cannot obtain in practice. However, we can apply traditional methods based on observed states, which result in the following baselines, as suggested in~\cite{li2018deepReinforce}:

% traditional optimization methods are nontrivial. The traditional baselines, as in~\cite{li2018deepReinforce}, include: 
 %model based methods divides resources equally \cite{li2018deepReinforce} with the following four level of methods. 
%We assume the   model based methods obtain different level of prior knowledge at the beginning of slicing time slots.
\begin{itemize}
    \item One-third equal allocation: the total bandwidth is equally sliced into three pieces.
    \item User-number-based allocation: the total bandwidth is sliced weighted by the number of users of each type. %In particular, assuming that the total bandwidth is $B$ and the total number of users is $U$, the allocated bandwidth is $\frac{u_{VoLTE}*B}{U}$, $\frac{u_{Video}*B}{U}$ and $\frac{u_{URLLC}*B}{U}$ respectively. This setting is same with the constrained RL.
    \item Packet-number-based allocation: the total bandwidth is sliced weighted by the number of packets of each type. 
    %Knowing the number of packet for each type , the total bandwidth is sliced based on the proportional of packet number. Assuming the packet number is $p_{VoLTE},p_{Video},p_{URLLC}$ respectively and the total packet number is $P$. Then the sliced action is $\frac{p_{VoLTE}*B}{P}$, $\frac{p_{Video}*B}{P}$ and $\frac{p_{URLLC}*B}{P}$.
    \item Traffic-demand-based allocation: the total bandwidth is sliced weighted by the current traffic demand. 
    %Knowing the exactly traffic demand for each type, the total bandwidth is sliced based on the proportional of traffic demand. Assuming the total traffic demand is $T$, then the slices are $\frac{t_{VoLTE}*B}{T}$, $\frac{t_{Video}*B}{T}$ and $\frac{t_{URLLC}*B}{T}$.
\end{itemize}
%Notice that the last two allocation approaches require more prior information than our algorithm and it is not available in practice before we make the allocation decision.
The RL-based methods include:
\begin{itemize}
    \item IPO: the original IPO~\cite{liu2020ipo} with fixed hyperparameter $t$.
    \item PPO: the state-of-the-art unconstrained RL algorithm, which usually outperforms DQN~\cite{li2018deepReinforce} and TRPO~\cite{schulman2015trust}.
    % the commonly applied unconstrained RL algorithm, which usually outperforms DQN~\cite{li2018deepReinforce} and TRPO~\cite{schulman2015trust}.
    \item PPO+Safelayer: the PPO method with Safelayer.
\end{itemize}
%we use the neural network to approximate the policy function, reward function and cost function given the high complexity of the action and state space in the network slicing setting. 
 We also compare with only adaptive IPO without Safelayer, to demonstrate the performance of Safelayer.
In the experiment, all policy neural-networks consist of two fully connected layers, with 64 and 32 nodes, respectively.

\input{figures.tex}

\subsection{Evaluation Results}

First, we demonstrate the evaluations with one single cumulative constraint selected from cumulative dissatisfaction ratio of Video, VoLTE and URLLC separately, in Fig.~\ref{fig:performance} (a-i). %For better demonstration, we use different y-axis scales.
%In these three figures, we are considering the dissatisfaction ratio of VoLTE, Video and URLLC as the cost of cumulative constraints respectively. 
Fig.~\ref{fig:Video_reward},~\ref{fig:Video_cumulative},~\ref{fig:VoLTE_reward},~\ref{fig:VoLTE_cumulative},~\ref{fig:URLLC_reward},~\ref{fig:URLLC_cumulative} show the results of long-term reward (throughput) and cumulative constraints (dissatisfaction ratio) with respect to the iterations of policy updates.  Both the rewards and constraints are cumulative values in $500$ time slots. 
For RL-based methods, the rewards and cumulative cost values update during the training process, while the traditional baselines do not improve or adapt to the changes in the environment. 
Even though PPO and PPO+Safelayer can achieve slightly higher reward than CLARA,  its cost value significantly violate the constraint. The original IPO violate the constraint as well. 
CLARA achieves fair high reward and satisfies cumulative constraints.
% because the hyperparameter $t$ is not well tuned.
%Because the traditional methods can not update the policy, they lack the ability of exploration and exploitation. So they remain the same over time. 
%In addition, the constrained RL achieves high cumulative reward remaining satisfied the cumulative constraints. Although the traffic demand based allocations get similar cumulative reward as constrained RL, the cumulative constraint is much higher than the limit. The one third based allocation methods are satisfied constrained limit, but the corresponding cumulative reward is low.

We also collect the policy after training 
and demonstrate the performance on  the implicit instantaneous constraints (latency) in $500$ time slots, compared with baselines, shown in 
Fig.~\ref{fig:Video_instantaneous},~\ref{fig:VoLTE_instantaneous},~\ref{fig:URLLC_instantaneous}.  
CLARA and PPO+Safelayer satisfy the latency requirements best of all. Adaptive IPO without the Safelayer are not guaranteed to satisfy the instantaneous constraint. 
Remark that in Fig.~\ref{fig:URLLC_reward} and \ref{fig:URLLC_cumulative}, we have tested an extreme case, where the traffic demand of URLLC is larger than the total bandwidth. 
%, the BS can't satisfy its constraint even allocate all bandwidth to it, detail will be discussed in section~\ref{sec:explanation}. 
The traffic demand based allocation benefits from knowing the extra information of the total traffic demand volume, 
however, our algorithm still performs almost the same. 
Our method can satisfy the latency constraint while the traffic demand allocation cannot. Moreover, since explicit instantaneous constraints (bandwidth allocation) can be numerically checked. CLARA can always make sure that they are satisfied.
% as well as the traffic demand based allocation. 
%Even the constrained RL based method didn't show too much advantage comparing to traffic demand based allocation, it perform as good as it. Traffic demand based allocation has a strong assumption that we know the exactly traffic demand for each type. We can get a optimal solution in each slicing time slot, which wouldn't waste any bandwidth. On the contrary, the observation of constrained RL is the number of users for each type which is weak than traffic demand.

In the setting of multiple cumulative constraints, 
we consider the cumulative constraints on the dissatisfaction ratio of Video and VoLTE users together. 
%first find a feasible policy under two cumulative constraints, then update the policy with Phase II as in algorithm \ref{alg:IPO}. 
The results in Fig.~\ref{fig:performance} (j-l) show that CLARA achieves the equivalent best reward as well as the satisfaction of the two cumulative constraints. 
All above, the final policy learned by CLARA outperforms all the baselines in either reward or constraint cost, if not both. 
 %From figures \ref{fig:VoLTE_instantaneous} and \ref{fig:Video_instantaneous}, we can see our constrained RL based method can reduce the latency to satisfy the instantaneous limit, which is $2s$, in most time. However, the   model based methods violates the constraint all the time. Even figure \ref{fig:URLLC_instantaneous} shows the URLLC latency violates the limit a lot, but the latency under constrained RL is still the lowest.
% The variation of the implicit instantaneous constraint is large because we do not have an accurate closed-form formulation and have to it learned using existing data.

%In addition, we are adding the softmax layer after the policy network. We can make sure that the explicit instantaneous constraints are satisfied.

%In summary, the constrained RL methods can get higher cumulative reward under satisfying the cumulative and instantaneous constraints. However, the   model based methods is either getting high reward but violating the constraints or satisfying the constraints but getting lower reward. Besides, the   model based methods lack the ability to learn the effect of the hidden structure although new data is obtained in each time slot.

%\subsection{Multiple Cumulative Constraints}\label{sec:multi}

%Our constrained RL algorithm can be extended to satisfy multiple cumulative constraints. 

%we can get policy satisfying both cumulative constraints and with high reward comparing to the   model based methods.

\subsection{Convergence speed}

%The convergence is a challenge in  reinforcement learning algorithms. 
In Fig.~\ref{fig:performance}, we have seen that in most cases CLARA converges within $100$ iterations. 
We can further speed up the convergence, 
by utilizing  domain knowledge and traditional baselines to obtain an initial policy other than a purely random one to warm start. 
As shown by ``Warmstart'' in Fig.~\ref{fig:performance} (a-c), we use the traffic demand based allocation as our initial policy, which has a high reward although a high cost violating the constraint;  
after a few iterations, the reward converges to the reward achieved starting with a random policy and the constraint is satisfied. 
%During the whole training process, the reward remains high.
%Furthermore, domain knowledge can be incorporated to effectively reduce the state/action space, which speeds up the convergence as well. 
The convergence in RL is an essential and challenging issue. A lot of potential future work need to be explored to speed up the convergence, e.g., better exploration techniques and model-based approaches, etc.

\begin{table}[t]
	\centering
	\caption{Statistics of three types of users}
	\begin{tabular}{|c| c c c|} 
		\hline
		& Video &VoLTE &  URLLC \\
		\hline
		Number of user& 50.15 &49.44 &  1.92  \\
		\hline
		Number of packet& 5950.56 &601.50 &  10.70  \\
		\hline
		Traffic Demand (kb)& 7002.35 &187.97 &  295290.37  \\
		\hline
	\end{tabular}
	\label{table:request}
% 	\vspace{-4mm}
\end{table}

\begin{table}[t]
	\centering
	\caption{Average bandwidth allocation with algorithms}
	\begin{tabular}{|c| c c c|}
		\hline
		(kb)& Video &VoLTE &  URLLC \\
		\hline
	    One third& 34133.33&34133.33&34133.33\\ 
	    \hline
	    User number&50589.71&49871.62 &1938.67\\
	    \hline
	    Packet number&92847.71& 9385.29&167.01\\
		 \hline
	    Traffic Demand&2370.53& 63.63&99965.83\\
	    \hline
		 IPO &25641.12 &8933.74 &67842.22 \\
		 \hline
		 Adaptive IPO&22152.22 & 5320.23 & 74927.54\\
		 \hline
		 PPO & 22141.71 & 11958.00 & 68300.28\\
		 \hline
		 PPO+Safelayer&19563.07 & 7711.20& 75125.81\\
		 \hline
	    CLARA& 27104.34&7372.04 &67923.62\\
		\hline
	\end{tabular}
	\label{table:allocation}
% 	\vspace{-4mm}
\end{table}

\subsection{Discussions}\label{sec:explanation}
We show the statistics of three types of users to illustrate the diverse characteristics of the three network slices, in Table~\ref{table:request} as suggested by~\cite{li2018deepReinforce}. 
 We also show the average resource allocation with all methods in Table~\ref{table:allocation}.  
%the allocation of total bandwidth ($100Mbs$) with our algorithm and baselines, shown in Table \ref{table:explanation}. 
%The constrained RL is following the converged policy of Sec.~\ref{sec:multi}.
%Analyzing the statistics, we have the following discussions. 
Traffic demands from URLLC network slice users dominate over $90\%$ of the total traffic demands, while the number of users and packets of URLLC is far less. 
Given such heterogeneous user demands,  the one third equally allocation method, user number and packet number based method 
allocate much more than enough bandwidth to Video and VoLTE users, resulting in losing the high volume traffic from the URLLC users.  
As for the traffic demand based method, it focuses on the demands of URLLC network slices and works well in reward maximization with the extra information of the real user demands. 
However, this allocation results in the dissatisfaction and leaving of the users of the other two network slices. 
In the extreme case as in Fig.~\ref{fig:URLLC_cumulative}, the resource is not enough to satisfy the URLLC users, and therefore, the URLLC users are also unsatisfied and leave. 
In this case, the allocation with traffic demand based method loses all the users, while our CLARA can adapt to this situation, and allocate more to the users from Video and VoLTE network slices with low latency requirement, to keep a high user participation level, achieving a better overall performance. The other RL-based methods can achieve a good reward performance, but violate the cumulative or instantaneous constraints.

%% file: figures.tex
\begin{figure*}[hbt!]
    \centering
    \begin{subfigure}[t]{0.272\textwidth}
         \centering
         \includegraphics[width=\textwidth]{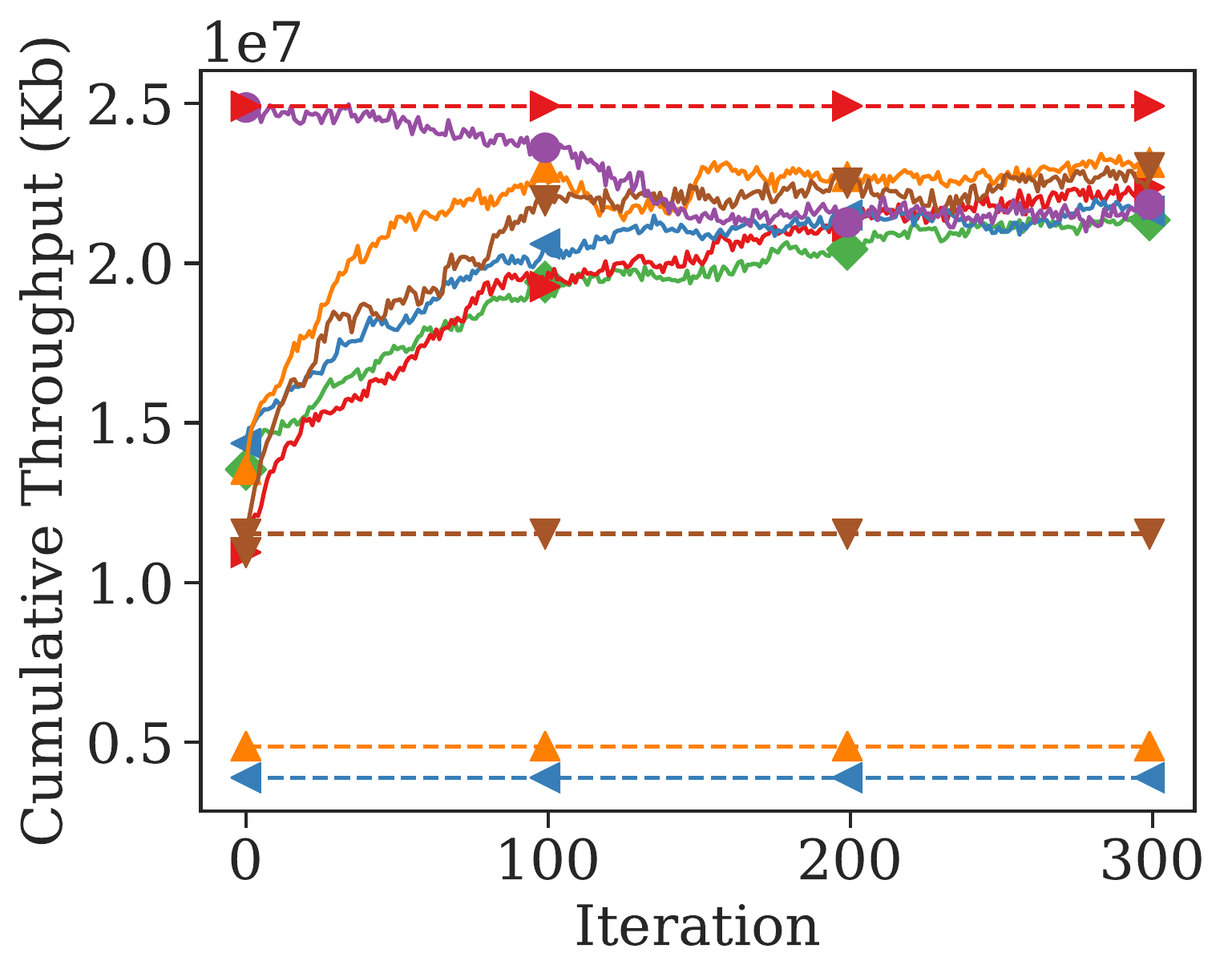}
         \caption{Reward under Video cumulative constraint}
         \label{fig:Video_reward}
     \end{subfigure}
     \hfill
     \begin{subfigure}[t]{0.272\textwidth}
         \centering
         \includegraphics[width=\textwidth]{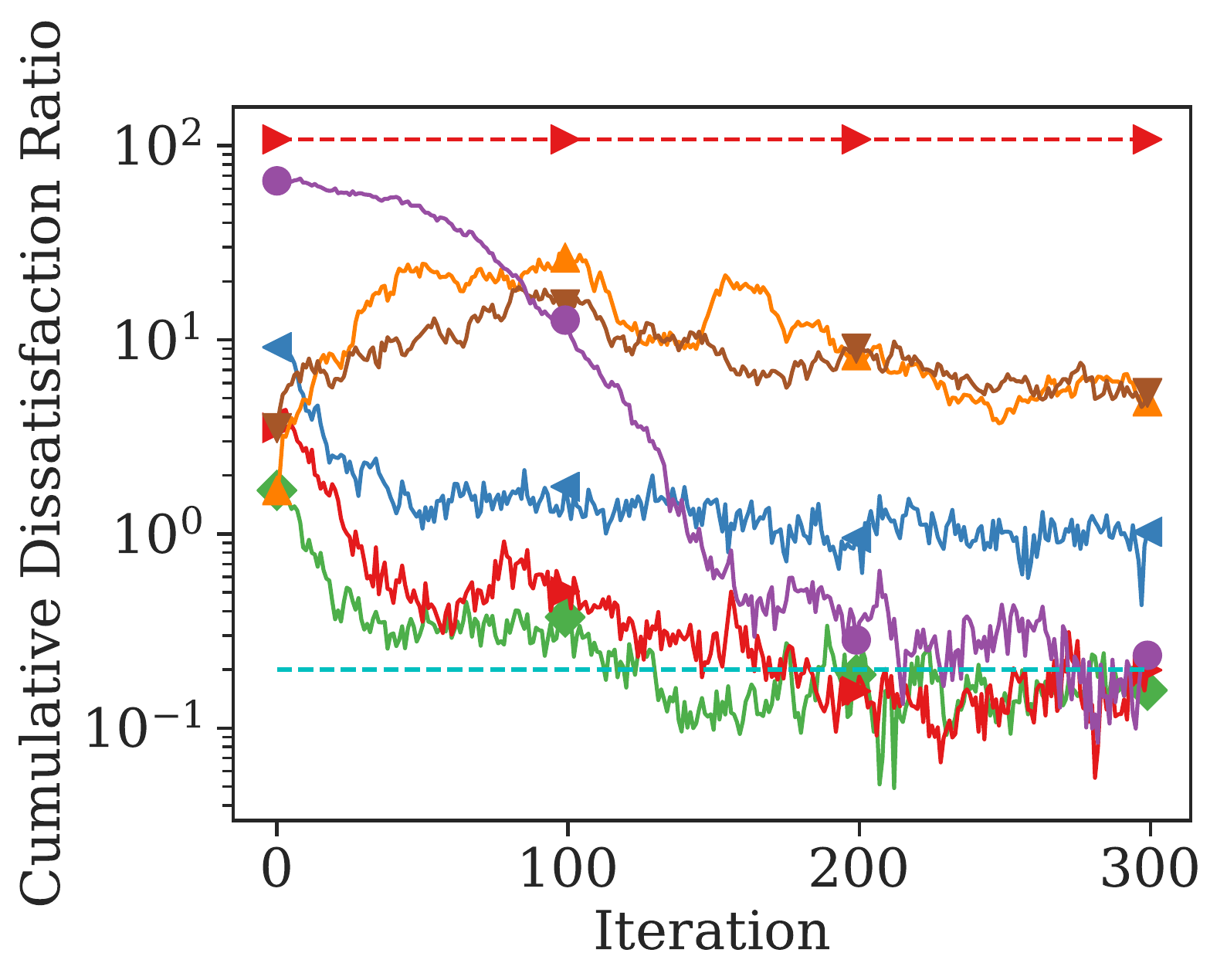}
         \caption{Video cumulative constraint (log-scaled y-axis)}
         \label{fig:Video_cumulative}
     \end{subfigure}
     \hfill
     \begin{subfigure}[t]{0.272\textwidth}
         \centering
         \includegraphics[width=\textwidth]{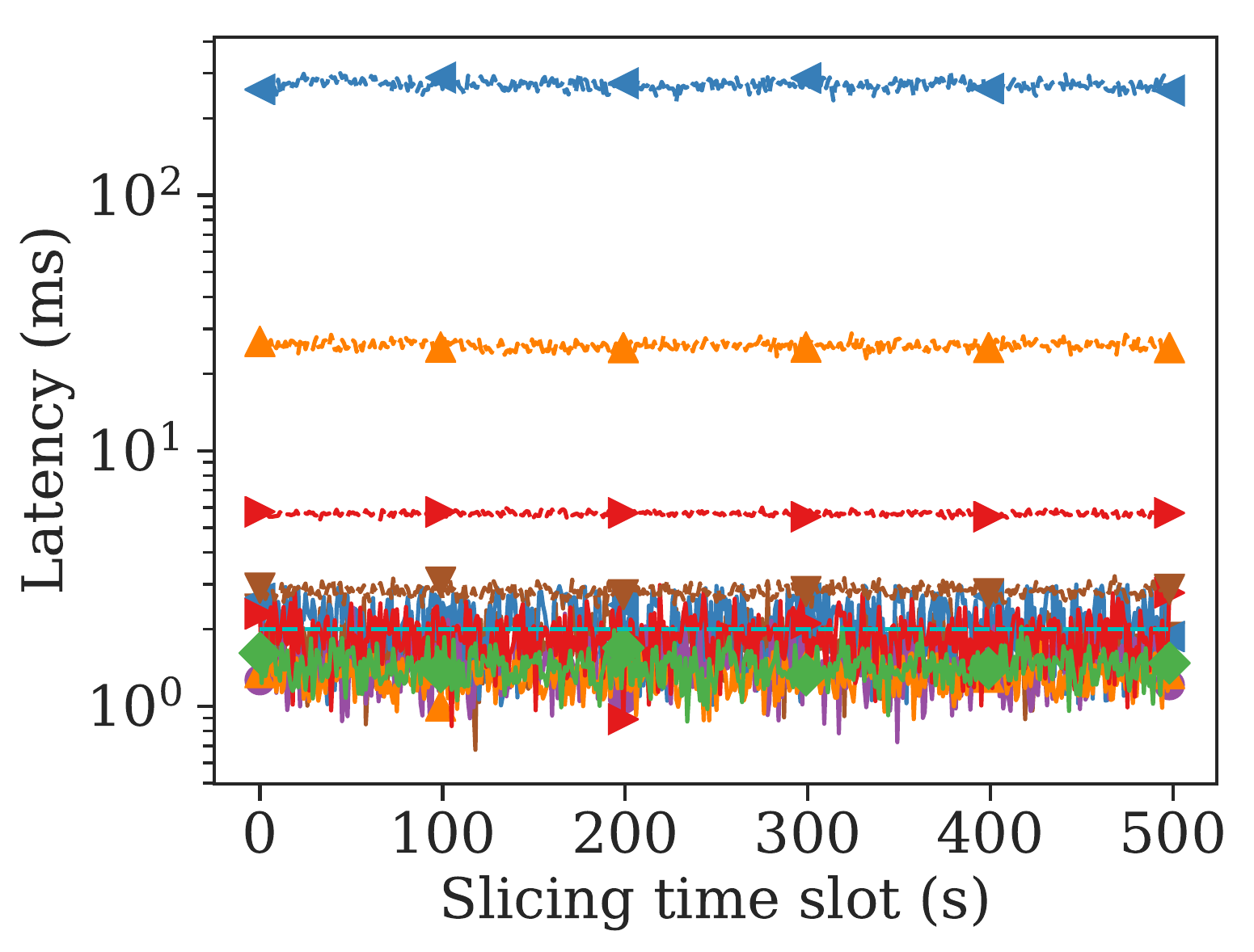}
         \caption{Instantaneous constraint under Video cumulative constraint (log-scaled y-axis)}
         \label{fig:Video_instantaneous}
     \end{subfigure}
     \hfill
     \begin{subfigure}[t]{0.272\textwidth}
         \centering
         \includegraphics[width=\textwidth]{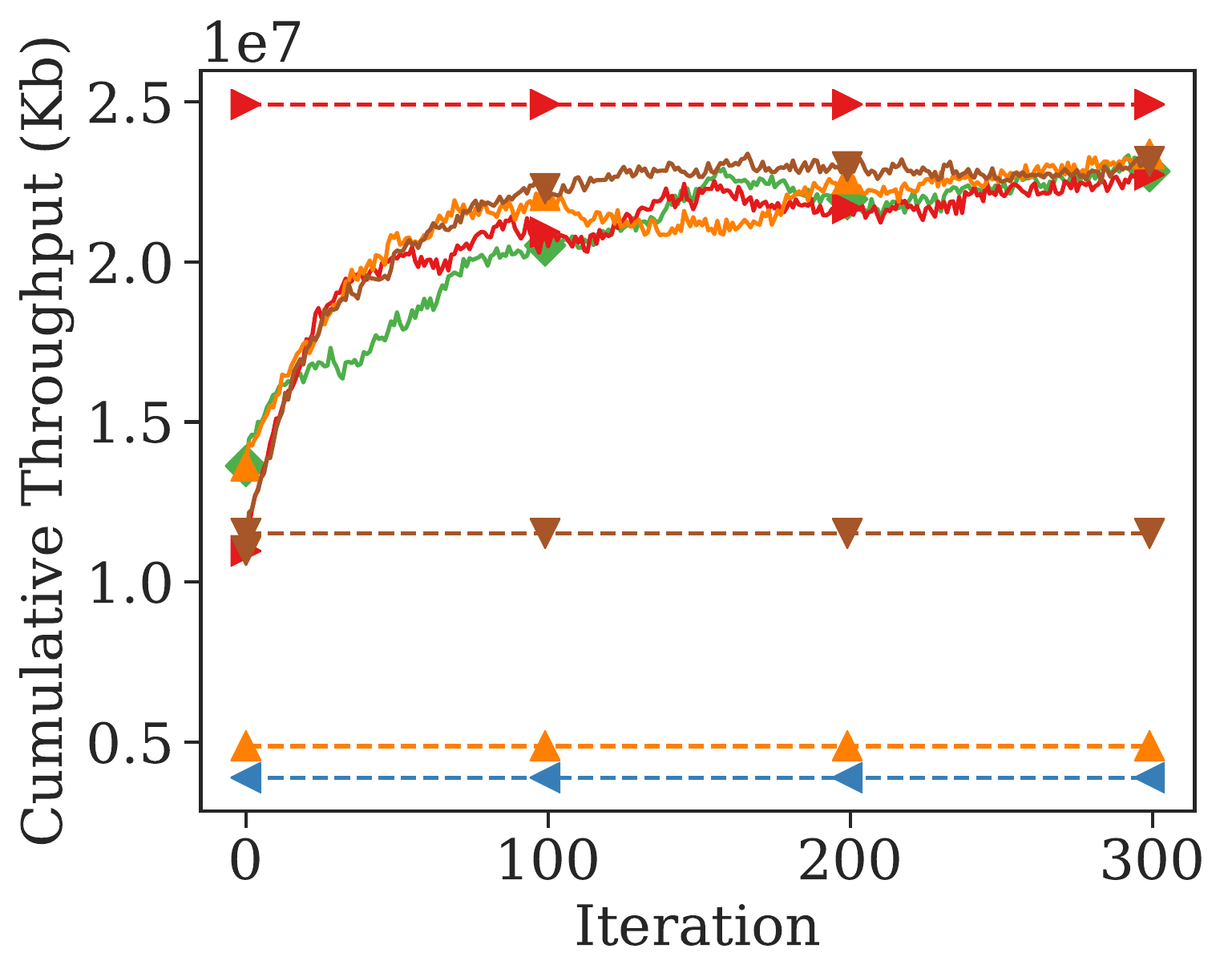}
         \caption{Reward under VoLTE cumulative constraint}
         \label{fig:VoLTE_reward}
     \end{subfigure}
     \hfill
     \begin{subfigure}[t]{0.272\textwidth}
         \centering
         \includegraphics[width=\textwidth]{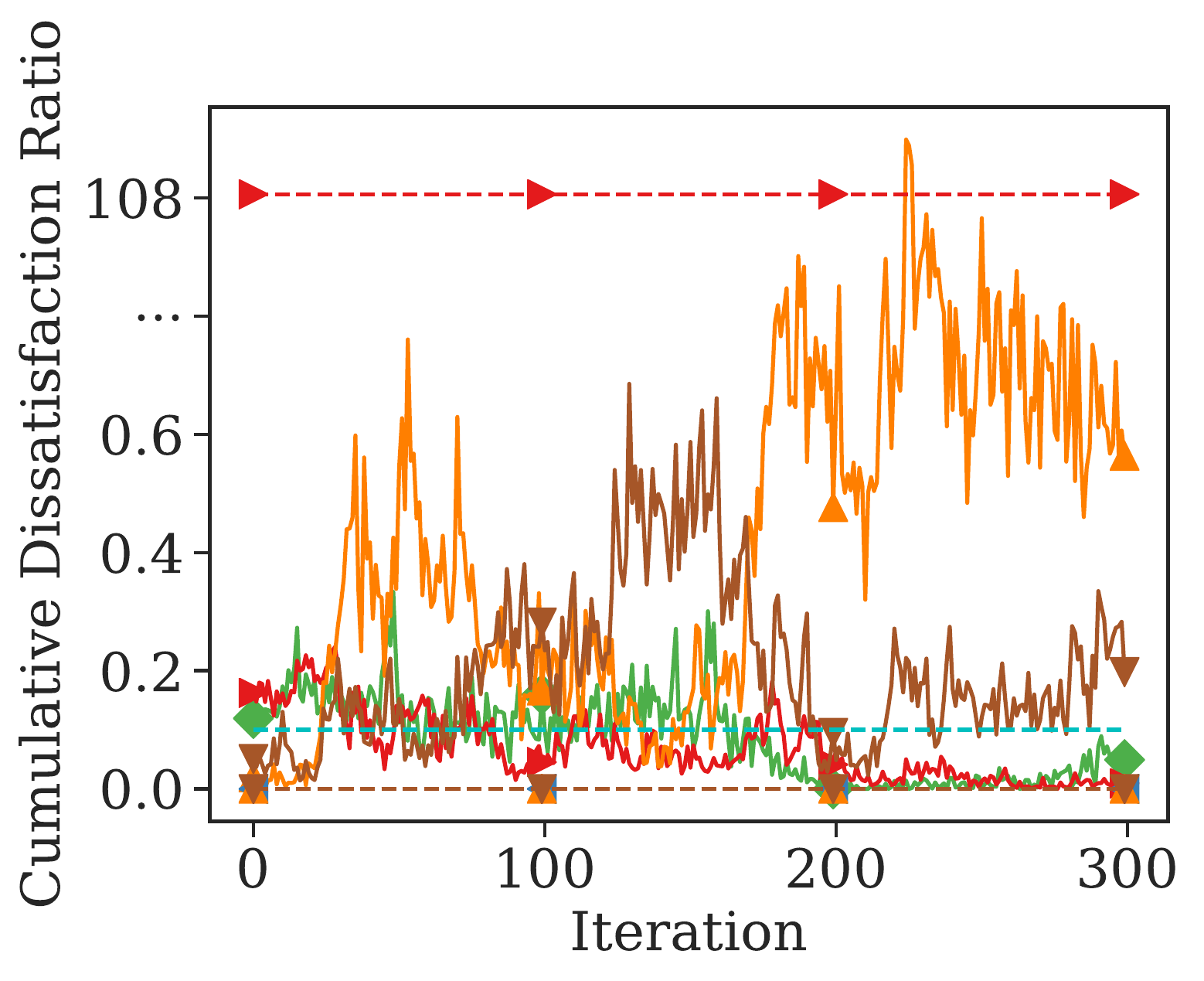}
         \caption{VoLTE cumulative constraint (nonlinear y-axis)}
         \label{fig:VoLTE_cumulative}
     \end{subfigure}
     \hfill
     \begin{subfigure}[t]{0.272\textwidth}
         \centering
         \includegraphics[width=\textwidth]{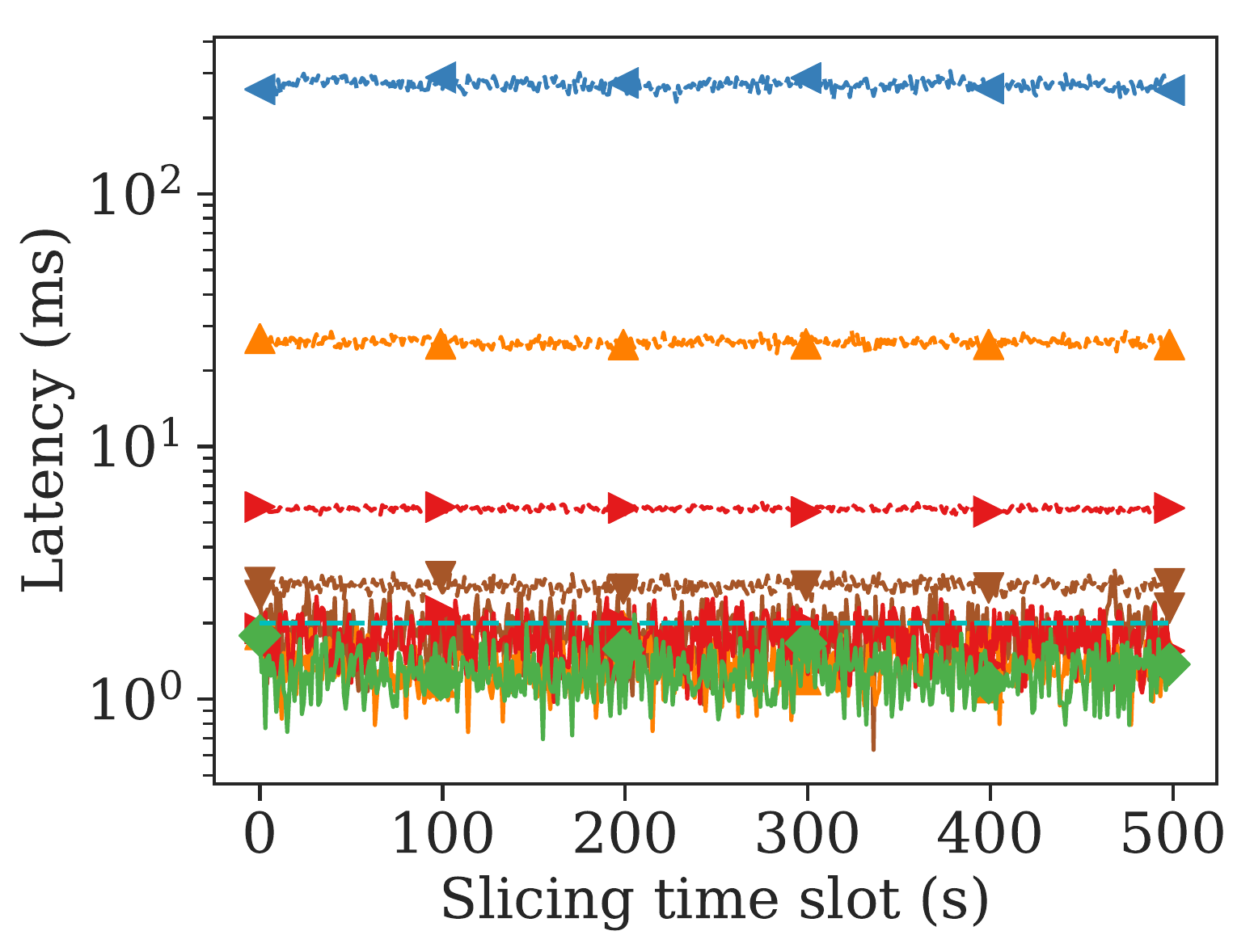}
         \caption{Instantaneous constraint under VoLTE cumulative constraint (log-scaled y-axis)}
         \label{fig:VoLTE_instantaneous}
     \end{subfigure}
     \hfill
     \begin{subfigure}[t]{0.272\textwidth}
         \centering
         \includegraphics[width=\textwidth]{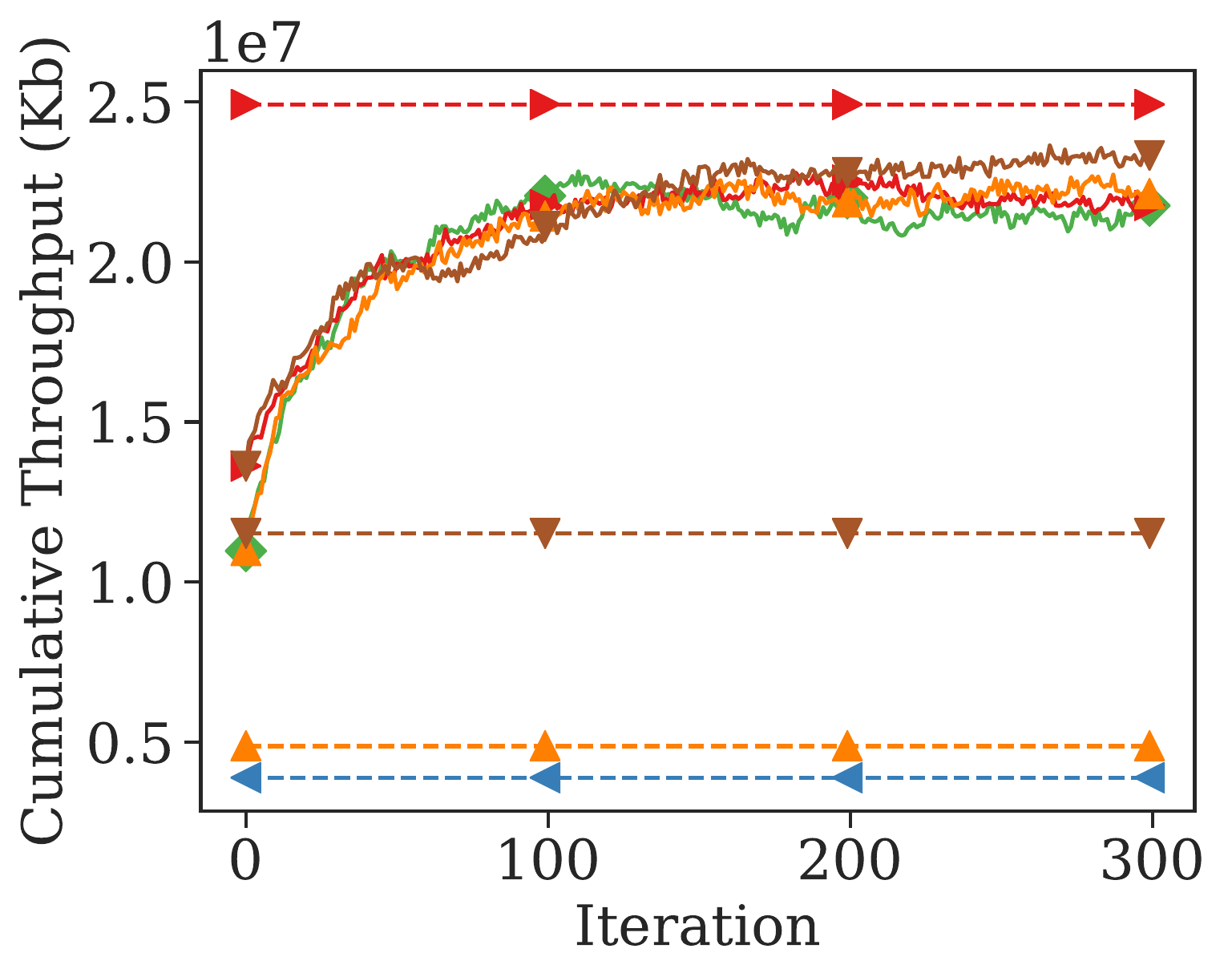}
         \caption{Reward under URLLC cumulative constraint}
         \label{fig:URLLC_reward}
     \end{subfigure}
     \hfill
     \begin{subfigure}[t]{0.272\textwidth}
         \centering
         \includegraphics[width=\textwidth]{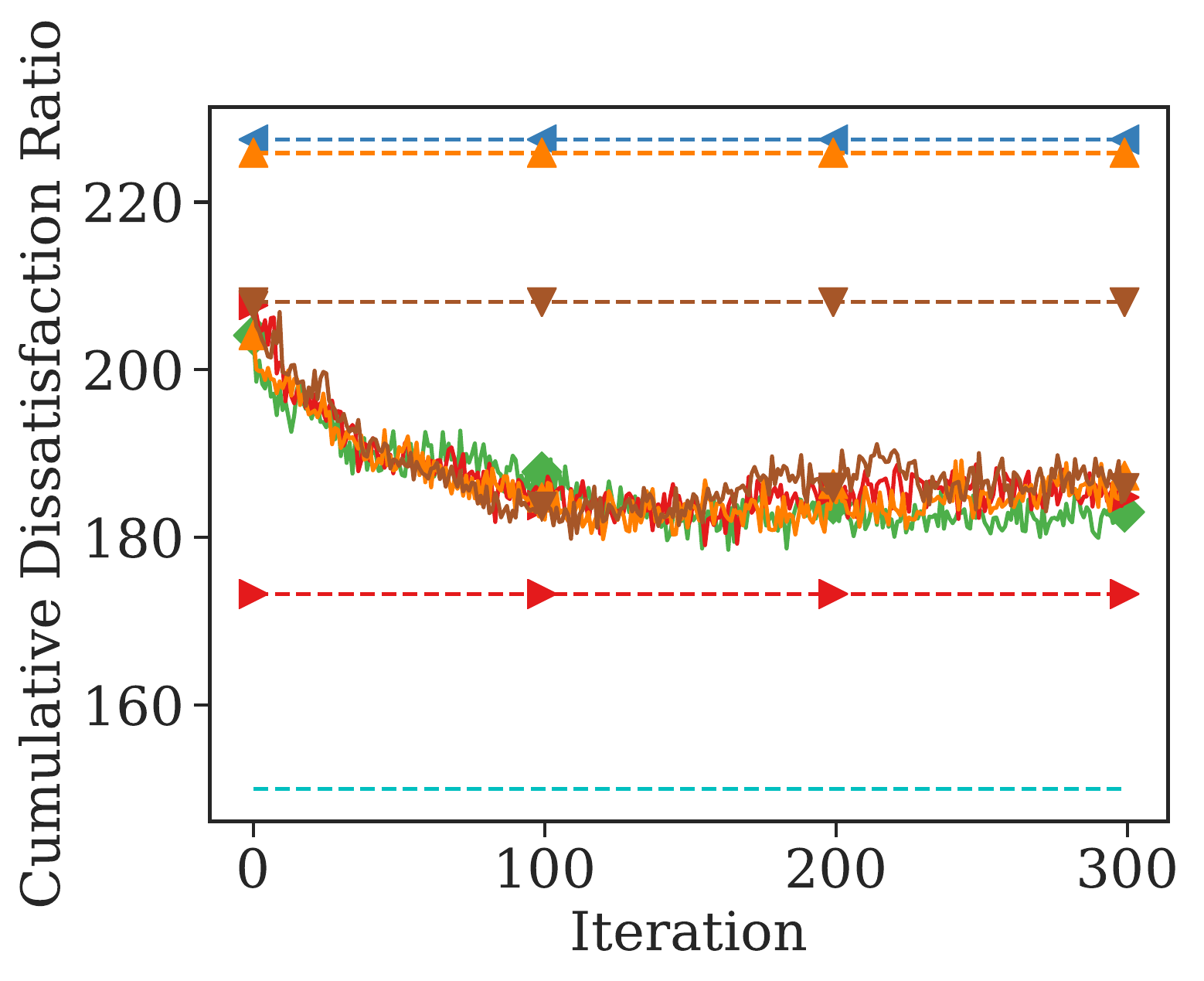}
         \caption{URLLC cumulative constraint}
         \label{fig:URLLC_cumulative}
     \end{subfigure}
      \hfill
     \begin{subfigure}[t]{0.272\textwidth}
         \centering
         \includegraphics[width=\textwidth]{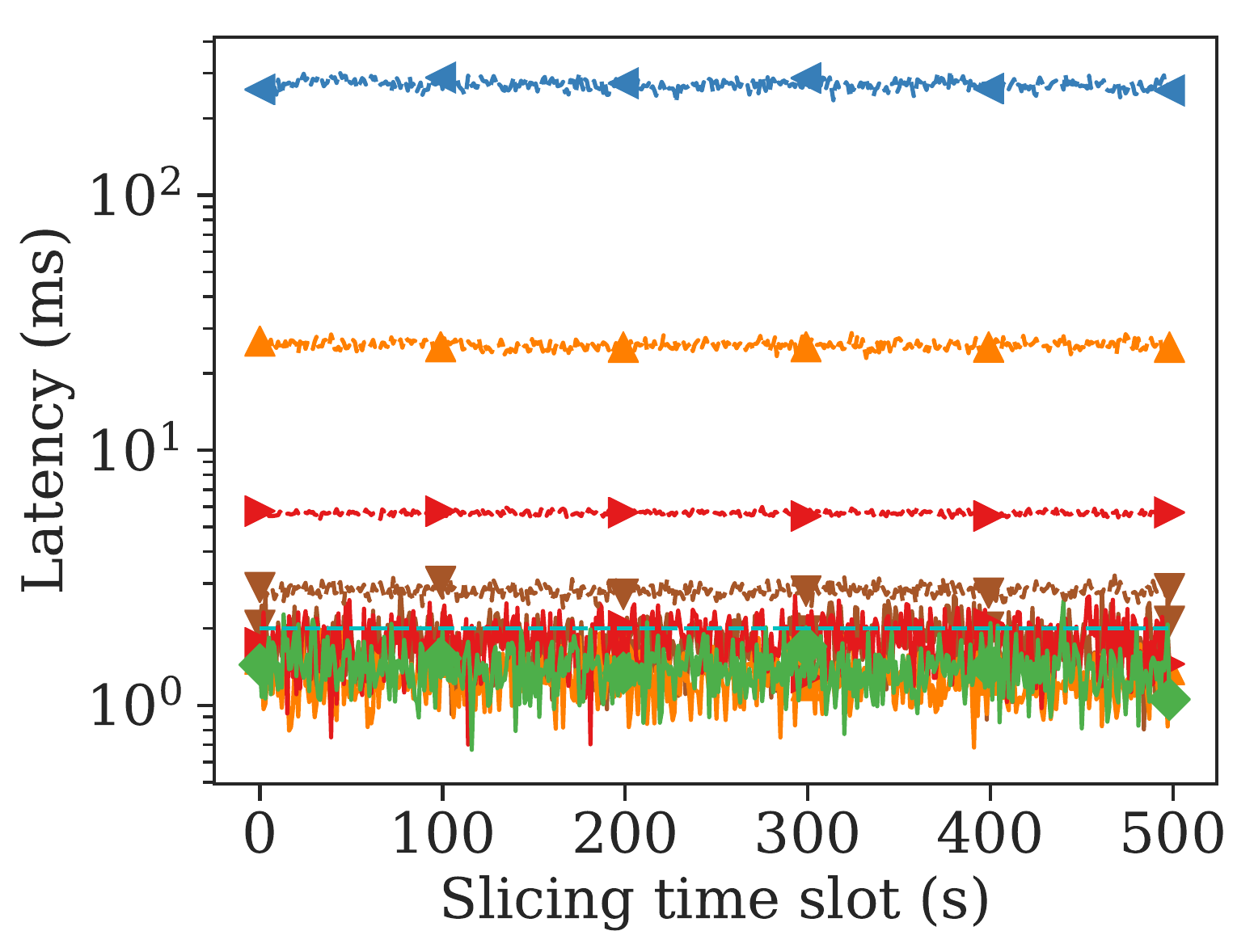}
         \caption{Instantaneous constraint under URLLC cumulative constraint (log-scaled y-axis)}
         \label{fig:URLLC_instantaneous}
     \end{subfigure}
     \hfill
     \begin{subfigure}[t]{0.272\textwidth}
         \centering
         \includegraphics[width=\textwidth]{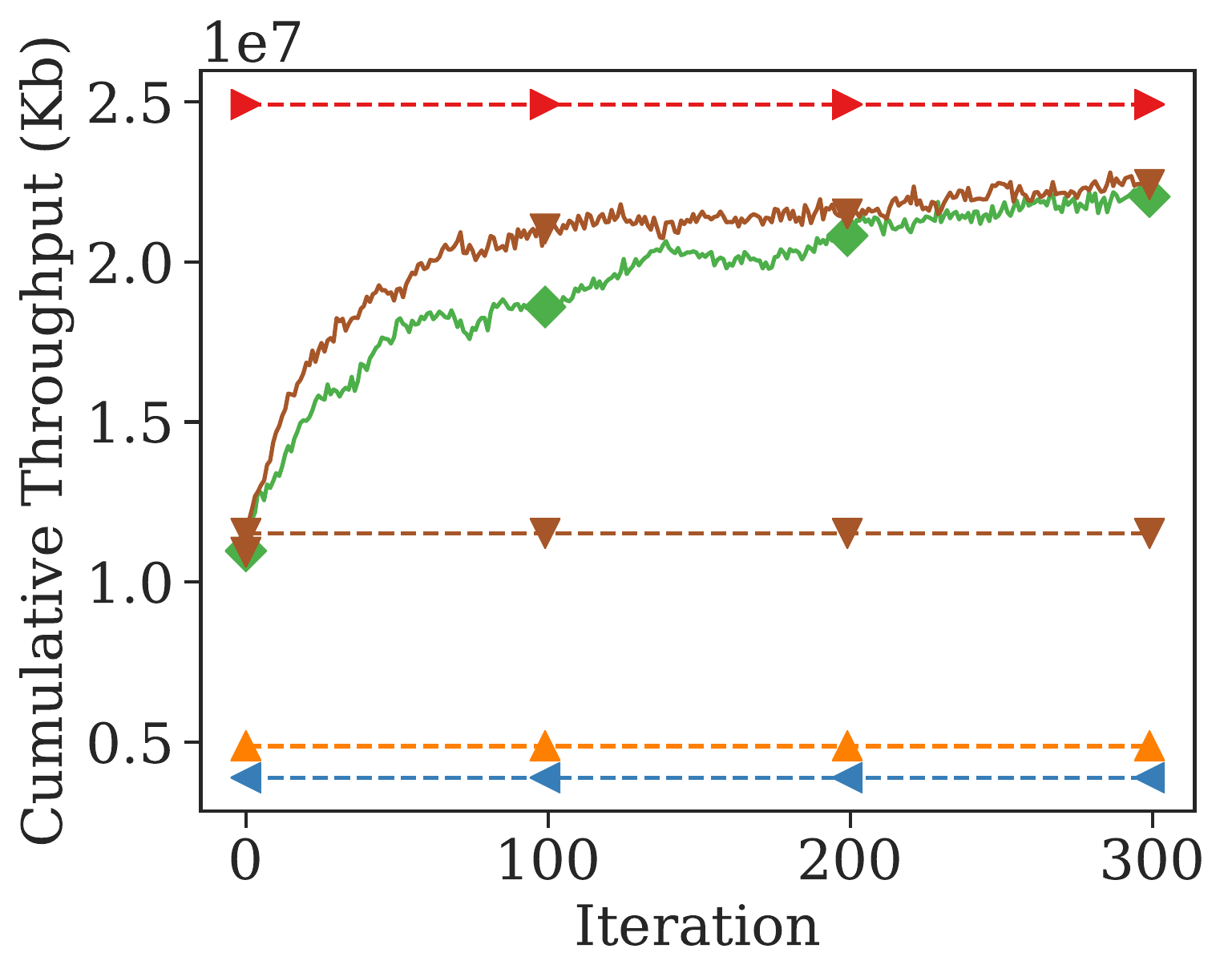}
         \caption{Reward under Video and VoLTE cumulative constraints}
         \label{fig:multi_reward}
     \end{subfigure}
     \hfill
     \begin{subfigure}[t]{0.272\textwidth}
         \centering
         \includegraphics[width=\textwidth]{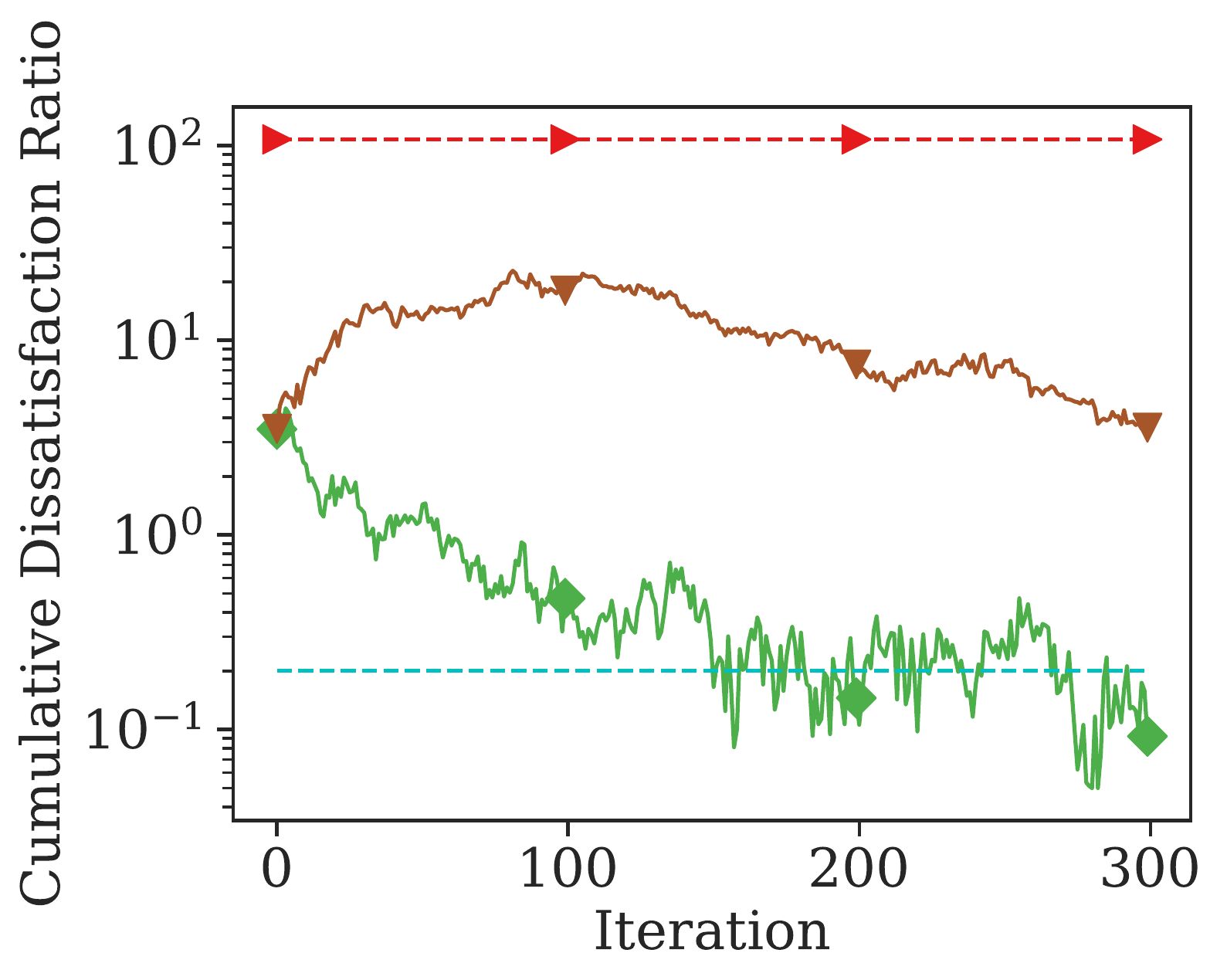}
         \caption{Video cumulative constraint (log-scaled y-axis)}
         \label{fig:multi_Video_cumulative}
     \end{subfigure}
     \hfill
     \begin{subfigure}[t]{0.272\textwidth}
         \centering
         \includegraphics[width=\textwidth]{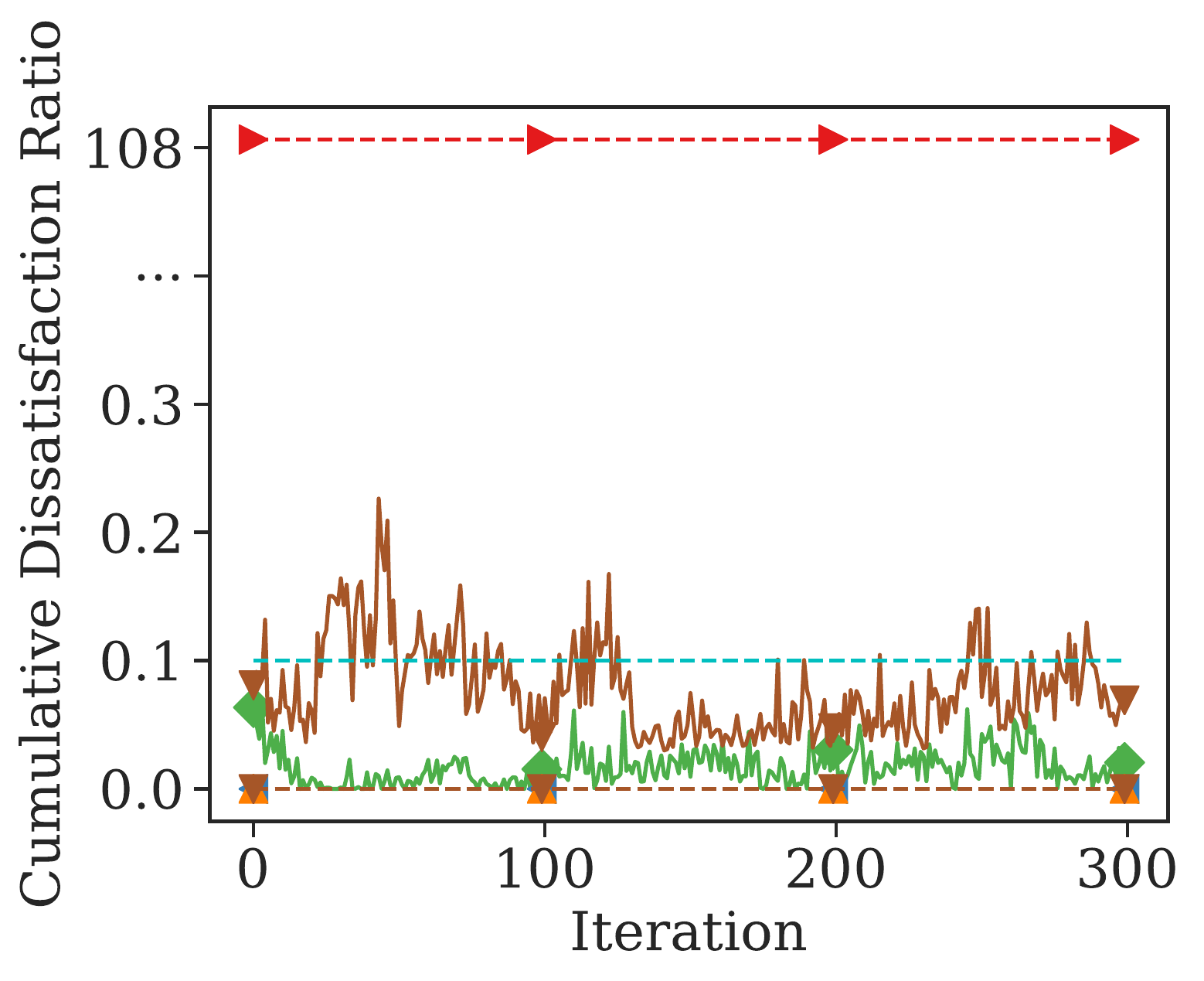}
         \caption{VoLTE cumulative constraint (nonlinear y-axis)}
         \label{fig:multi_VoLTE_cumulative}
     \end{subfigure}
    \begin{minipage}[t]{0.9\textwidth}
        \centering
        %  \vspace{2mm}
        \includegraphics[width=\textwidth]{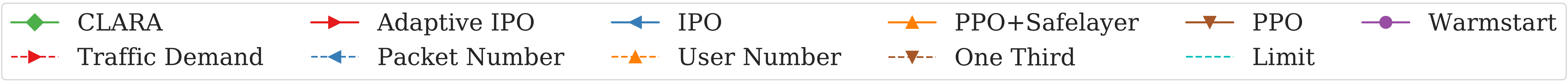}
    \end{minipage}
     \caption{Average performance under different cumulative constraints:  Fig.~\ref{fig:Video_reward},~\ref{fig:Video_cumulative} and \ref{fig:Video_instantaneous} are under single Video cumulative constraint; Fig.~\ref{fig:VoLTE_reward},~\ref{fig:VoLTE_cumulative} and \ref{fig:VoLTE_instantaneous} are under single VoLTE cumulative constraint; Fig.~\ref{fig:URLLC_reward},~\ref{fig:URLLC_cumulative} and \ref{fig:URLLC_instantaneous} are under single URLLC cumulative constraint; Fig.~\ref{fig:multi_reward},~\ref{fig:multi_Video_cumulative} and \ref{fig:multi_VoLTE_cumulative} are under both Video and VoLTE cumulative constraints; the dash lines are limits for different constraints. In Fig.~\ref{fig:Video_cumulative},~\ref{fig:multi_Video_cumulative}, the cumulative values for Packet Number, User Number, and One Third are $0$.
     %; zero values do not display on logarithmic scale y-axis.
     }
     \label{fig:performance}
    %   \vspace{-4mm}
\end{figure*}

%% file: related.tex
\section{Related work}
\subsection{Network Slicing}
Rost et al. in~\cite{rost2017network} propose that 
network slicing, which enables the multiplexing of virtualized and independent logical networks on the same physical network infrastructure, 
is an efficient solution that addresses the diverse requirements of 5G mobile networks, through analyzing the realization options of a flexible radio access network slicing. 
Different service models and architectures are suggested, such as those in~\cite{ taleb2017multi}. 
The key of network slicing is efficient resource allocation. The dynamic resource allocation in network slicing can improve the resource efficiency~\cite{marquez2018should}. 
However, resource allocation in network slicing is proved to be NP-hard~\cite{zhang2017network,d2019slice} and heuristic algorithms are proposed, while the performances are not guaranteed.

\subsection{Reinforcement Learning in Network Slicing}
%Machine learning has been widely used in networks, including using RL to solve network slicing problems. 
%Typically, RL deals with three basic elements: state, action and reward. 
% It's challenging to map network slicing problem to RL. 
RL is  employed  to allocate resources in network slicing.  
In~\cite{li2018deepReinforce}, the authors study the setting of demand-aware network slicing. 
Network bandwidth is allocated to three types of slices: Video, VoLTE, and URLLC; 
the state is the traffic load in each slice; the action is the bandwidth allocation; the reward is the weighted sum of spectrum efficiency (SE) and QoE of the slices. 
Bega et al.~\cite{bega2017optimising} study admission control of two types of network slice requests with different price and service requirements. 
The state is the number of current network slice users in the system; the action is the admission decision of the provider; and the reward is the revenue of the service provider.  
In~\cite{chen2018optimized}, the system allocates the RAN and mobile-edge computing (MEC) slices to provide computation offloading services to mobile users; the state consists of task queue, energy queue, and channel qualities;
the action is the decision to execute the offloading computation; 
the reward is the total utility value of all users. 
Similar works can be found in~\cite{sciancalepore2019rl,hua2019gan}. These works  perform efficiently, however, they do not consider the crucial constraints  in  network  slicing usually require knowing  the traffic demand or user mobility which is not realistic in most cases.
%While there exists some work on reinforcement learning based network slicing and constrained reinforcement learning, 
A very recent work~\cite{xu2021constrained} considers constraints in network slicing. However, the Lagrangian Relaxation method they are using to handle the constraints is not as efficient as our IPO~\cite{liu2020ipo} where CLARA relies on.
An earlier short version of our work appears at a workshop~\cite{liu2020constrained} and another short version appeared as a poster~\cite{liu2021resource}. In comparison, in this paper, we expand the system architecture and description, we prove the policy improvement in Theorem~\ref{thm:PolicyImprove}, we show CLARA can handle multiple cumulative constraints that are more common in real-world scenarios, we demonstrate the feasibility of warm-start to speed up convergence, and we expand the evaluation to compare with more baselines. 

% Our earlier work was accepted by the HDR-Nets workshop. The anonymous paper is posted on~\cite{liu2020clara}. However, the previous work does not clearly illustrate the system architecture with CLARA and more benefit of CLARA was left as a future work. In this paper, we emphasize the system description in Section~\ref{Sec:sys} and proof policy improvement in Theorem~\ref{thm:PolicyImprove}. We also compare CLARA with more RL-based baselines in Section~\ref{sec:experiments} to demonstrate the affect of Adaptive IPO and Safelayer separately. Moreover, we show CLARA can deal with multiple  cumulative constraints which is common in real-world. To the best of our knowledge, we are the first to apply constrained RL to solve network slicing problems with constraints.

\subsection{Constrained Reinforcement Learning}
Although much progress has been made in RL, while the work on constrained RL is limited~\cite{liu2021policy,liu2021cts2}. %The constraints in RL can include cumulative and instantaneous constraints. %The instantaneous constraints can be furthermore divided into explicit and implicit constraints.
The most common approach is to use Lagrangian relaxation~\cite{chow2017risk,tessler2018reward}. It reduces the constrained optimization problem to an unconstrained one by adding constraint as a weighted penalty to the objective function,
%The weight balances the tradeoff between the objective function and the constraint is the Lagrange multipliers.
%Primal-Dual Optimization (PDO) \cite{chow2017risk} is the first work applying Lagrangian relaxation and policy gradient methods to solve risk-constrained RL. In this approach, it is difficult to find the appropriate Lagrange multipliers. To accelerate the training process, \cite{liang2018accelerated} proposed to use off-line data pre-train the Lagrange multipliers, then update the policy with online data. Moreover, a batch policy learning based method \cite{le2019batch} devises an algorithm to learn constrained policy with batch off-line data, which can solve the constrained problem while improving sampling efficiency. Reward Constrained Policy Optimization (RCPO) \cite{tessler2018reward} proposes an approach based on TD-learning to update the policy which accelerates the learning process.
however, Lagrangian relaxation methods are sensitive to the initialization of the Lagrange multipliers and the learning rate and  constraint satisfaction is only guaranteed upon convergence.
Constrained Policy Optimization (CPO)~\cite{achiam2017constrained} employs TRPO to approximate the complex constrained optimization problem with a quadratic optimization. However, CPO can result in the high computation training cost in large scale problems and 
% CPO can not handle more general constraint settings such as mean valued constraints~\cite{tessler2018reward} and multiple constraints. 
Interior-point Policy Optimization (IPO)~\cite{liu2020ipo} takes advantage of the logarithmic barrier function to handle  cumulative constraints and achieves good performance, while IPO does not have theoretical guarantee for the policy improvement and the hyperparameters are not adaptive tuned. 
%in RL. Comparing to the state-of-art constrained RL, IPO achieves higher long-term reward, lower cumulative constraint values, and smaller variation, and thus a suitable candidate for network slicing. However, IPO does not consider instantaneous constraints, nor the initialization to a feasible point. Therefore, in this work, we adopt and adapt IPO to study network slicing.

%Lyapunov functions are also used for safe approximation~\cite{khalil2002nonlinear,neely2010stochastic}. However, different tasks need different Lyapunov functions that are difficult to design.

As for instantaneous constraints, a natural approach is to project the solutions to the feasible space. 
In~\cite{bhatia2019resource,dalal2018safe}, the authors propose approaches based on Deep Deterministic Policy Gradient (DDPG) and project the output of the infeasible actors neural network to the feasible space by adding a projection safety layer., 
% to handle different types of explicit constraints. 
%, global sum constraint, local minimum and maximum constraint and regional minimum and maximum constraint. 
%They project the output of the actor neural network to the feasible space of actions by adding a safety layer to the policy network. 
% The similar idea is applied to implicit instantaneous constraints in~\cite{dalal2018safe}. %Moreover, they train another neural network in advance to learn the value of implicit instantaneous constraints.

%Network slicing imposes multiple types of and multiple constraints on different network slices. Furthermore, policy performance during training matters. Thus, existing algorithms do not suffice the need of network slicing.

%% file: main.bbl
\begin{thebibliography}{10}
\providecommand{\url}[1]{#1}
\csname url@samestyle\endcsname
\providecommand{\newblock}{\relax}
\providecommand{\bibinfo}[2]{#2}
\providecommand{\BIBentrySTDinterwordspacing}{\spaceskip=0pt\relax}
\providecommand{\BIBentryALTinterwordstretchfactor}{4}
\providecommand{\BIBentryALTinterwordspacing}{\spaceskip=\fontdimen2\font plus
\BIBentryALTinterwordstretchfactor\fontdimen3\font minus
  \fontdimen4\font\relax}
\providecommand{\BIBforeignlanguage}[2]{{%
\expandafter\ifx\csname l@#1\endcsname\relax
\typeout{** WARNING: IEEEtran.bst: No hyphenation pattern has been}%
\typeout{** loaded for the language `#1'. Using the pattern for}%
\typeout{** the default language instead.}%
\else
\language=\csname l@#1\endcsname
\fi
#2}}
\providecommand{\BIBdecl}{\relax}
\BIBdecl

\bibitem{lu2015safeguard}
L.~Lu and Y.~Liu, ``Safeguard: User reauthentication on smartphones via
  behavioral biometrics,'' \emph{IEEE Transactions on Computational Social
  Systems}, vol.~2, no.~3, pp. 53--64, 2015.

\bibitem{qu2015improved}
H.~Qu, X.~Xie, Y.~Liu, M.~Zhang, and L.~Lu, ``Improved perception-based spiking
  neuron learning rule for real-time user authentication,''
  \emph{Neurocomputing}, vol. 151, pp. 310--318, 2015.

\bibitem{lu2017sense}
L.~Lu, L.~Liu, M.~J. Hussain, and Y.~Liu, ``I sense you by breath: Speaker
  recognition via breath biometrics,'' \emph{IEEE Transactions on Dependable
  and Secure Computing}, vol.~17, no.~2, pp. 306--319, 2017.

\bibitem{liu2018less}
Y.~Liu, J.~Chen, and H.~Chen, ``Less is more: Culling the training set to
  improve robustness of deep neural networks,'' in \emph{International
  Conference on Decision and Game Theory for Security}.\hskip 1em plus 0.5em
  minus 0.4em\relax Springer, 2018, pp. 102--114.

\bibitem{3gpp2015cellular}
3rd Generation Partnership Projec~(3GPP), ``Cellular system support for ultra
  low complexity and low throughput {I}nternet of {T}hings,'' \emph{3GPP
  Technical Report (TR) 45.820}, 2015.

\bibitem{li2017review}
X.~Li, D.~Li, J.~Wan, A.~V. Vasilakos, C.-F. Lai, and S.~Wang, ``A review of
  industrial wireless networks in the context of industry 4.0,'' \emph{Wireless
  networks}, vol.~23, no.~1, pp. 23--41, 2017.

\bibitem{foukas2017network}
X.~Foukas, G.~Patounas, A.~Elmokashfi, and M.~K. Marina, ``Network slicing in
  5{G}: {S}urvey and challenges,'' \emph{IEEE Communications Magazine},
  vol.~55, no.~5, pp. 94--100, 2017.

\bibitem{alliance20155g}
N.~Alliance, ``5{G} white paper,'' \emph{Next generation mobile networks, white
  paper}, vol.~1, 2015.

\bibitem{nikaein2015network}
N.~Nikaein, E.~Schiller, R.~Favraud, K.~Katsalis, D.~Stavropoulos, I.~Alyafawi,
  Z.~Zhao, T.~Braun, and T.~Korakis, ``Network store: Exploring slicing in
  future 5{G} networks,'' in \emph{Proceedings of the 10th International
  Workshop on Mobility in the Evolving Internet Architecture}.\hskip 1em plus
  0.5em minus 0.4em\relax ACM, 2015, pp. 8--13.

\bibitem{zhang2017network}
H.~Zhang, N.~Liu, X.~Chu, K.~Long, A.-H. Aghvami, and V.~C. Leung, ``Network
  slicing based 5{G} and future mobile networks: mobility, resource management,
  and challenges,'' \emph{IEEE Communications Magazine}, vol.~55, no.~8, pp.
  138--145, 2017.

\bibitem{itfuture}
``{ITU} workshop on {Machine Learning for 5{G} and beyond},'' 2018.

\bibitem{mao2017neural}
H.~Mao, R.~Netravali, and M.~Alizadeh, ``Neural adaptive video streaming with
  pensieve,'' in \emph{Proceedings of the Conference of the ACM Special
  Interest Group on Data Communication}.\hskip 1em plus 0.5em minus 0.4em\relax
  ACM, 2017, pp. 197--210.

\bibitem{uzakgider2015learning}
T.~Uzakgider, C.~Cetinkaya, and M.~Sayit, ``Learning-based approach for layered
  adaptive video streaming over sdn,'' \emph{Computer Networks}, vol.~92, pp.
  357--368, 2015.

\bibitem{mao2016resource}
H.~Mao, M.~Alizadeh, I.~Menache, and S.~Kandula, ``Resource management with
  deep reinforcement learning,'' in \emph{Proceedings of the 15th ACM Workshop
  on Hot Topics in Networks}.\hskip 1em plus 0.5em minus 0.4em\relax ACM, 2016,
  pp. 50--56.

\bibitem{ChuaiInfocom2019}
J.~Chuai, Z.~Chen, G.~Liu, X.~Guo, X.~Wang, X.~Liu, C.~Zhu, and F.~Shen, ``A
  collaborative learning based approach for parameter configuration of cellular
  networks,'' in \emph{IEEE INFOCOM 2019-IEEE Conference on Computer
  Communications}.\hskip 1em plus 0.5em minus 0.4em\relax IEEE, 2019, pp.
  1396--1404.

\bibitem{BaoBigData2016}
Y.~Bao, H.~Wu, T.~Zhang, A.~A. Ramli, and X.~Liu, ``Shooting a moving target:
  Motion-prediction-based transmission for 360-degree videos,'' in \emph{IEEE
  International Conferences on Big Data (IEEE BigData 2016)}, Dec 2016.

\bibitem{zhang2019macs}
Z.~Zhang, L.~Ma, K.~Poularakis, K.~K. Leung, J.~Tucker, and A.~Swami, ``Macs:
  Deep reinforcement learning based sdn controller synchronization policy
  design,'' in \emph{2019 IEEE 27th International Conference on Network
  Protocols (ICNP)}.\hskip 1em plus 0.5em minus 0.4em\relax IEEE, 2019, pp.
  1--11.

\bibitem{sengupta2018hotdash}
S.~Sengupta, N.~Ganguly, S.~Chakraborty, and P.~De, ``Hotdash: Hotspot aware
  adaptive video streaming using deep reinforcement learning,'' in \emph{2018
  IEEE 26th International Conference on Network Protocols (ICNP)}.\hskip 1em
  plus 0.5em minus 0.4em\relax IEEE, 2018, pp. 165--175.

\bibitem{liu2020ipo}
Y.~Liu, J.~Ding, and X.~Liu, ``Ipo: Interior-point policy optimization under
  constraints,'' in \emph{Proceedings of the AAAI Conference on Artificial
  Intelligence}, vol.~34, no.~04, 2020, pp. 4940--4947.

\bibitem{dalal2018safe}
G.~Dalal, K.~Dvijotham, M.~Vecerik, T.~Hester, C.~Paduraru, and Y.~Tassa,
  ``Safe exploration in continuous action spaces,'' \emph{arXiv preprint
  arXiv:1801.08757}, 2018.

\bibitem{bhatia2019resource}
A.~Bhatia, P.~Varakantham, and A.~Kumar, ``Resource constrained deep
  reinforcement learning,'' in \emph{Proceedings of the International
  Conference on Automated Planning and Scheduling}, vol.~29, no.~1, 2019, pp.
  610--620.

\bibitem{etsi2015network}
N.~ETSI, ``Network functions virtualization (nfv) infrastructure overview,''
  \emph{NFV-INF}, vol.~1, p.~V1, 2015.

\bibitem{URLLC}
\BIBentryALTinterwordspacing
NGMN, ``Ngmn radio access performance evaluation methodology,'' 2007. [Online].
  Available:
  \url{https://www.ngmn.org/publications/ngmn-radio-access-performance-evaluation-methodology.html}
\BIBentrySTDinterwordspacing

\bibitem{li2018deepReinforce}
R.~Li, Z.~Zhao, Q.~Sun, I.~Chih-Lin, C.~Yang, X.~Chen, M.~Zhao, and H.~Zhang,
  ``Deep reinforcement learning for resource management in network slicing,''
  \emph{IEEE Access}, vol.~6, pp. 74\,429--74\,441, 2018.

\bibitem{kullback1951information}
S.~Kullback and R.~A. Leibler, ``On information and sufficiency,'' \emph{The
  annals of mathematical statistics}, vol.~22, no.~1, pp. 79--86, 1951.

\bibitem{sutton2000policy}
R.~S. Sutton, D.~A. McAllester, S.~P. Singh, and Y.~Mansour, ``Policy gradient
  methods for reinforcement learning with function approximation,'' in
  \emph{Advances in neural information processing systems}, 2000, pp.
  1057--1063.

\bibitem{schulman2015trust}
J.~Schulman, S.~Levine, P.~Abbeel, M.~Jordan, and P.~Moritz, ``Trust region
  policy optimization,'' in \emph{International conference on machine
  learning}, 2015, pp. 1889--1897.

\bibitem{achiam2017constrained}
J.~Achiam, D.~Held, A.~Tamar, and P.~Abbeel, ``Constrained policy
  optimization,'' in \emph{Proceedings of the 34th International Conference on
  Machine Learning-Volume 70}.\hskip 1em plus 0.5em minus 0.4em\relax JMLR.
  org, 2017, pp. 22--31.

\bibitem{schulman2017proximal}
J.~Schulman, F.~Wolski, P.~Dhariwal, A.~Radford, and O.~Klimov, ``Proximal
  policy optimization algorithms,'' \emph{arXiv preprint arXiv:1707.06347},
  2017.

\bibitem{boyd2004convex}
S.~Boyd and L.~Vandenberghe, \emph{Convex optimization}.\hskip 1em plus 0.5em
  minus 0.4em\relax Cambridge university press, 2004.

\bibitem{rost2017network}
P.~Rost, C.~Mannweiler, D.~S. Michalopoulos, C.~Sartori, V.~Sciancalepore,
  N.~Sastry, O.~Holland, S.~Tayade, B.~Han, D.~Bega \emph{et~al.}, ``Network
  slicing to enable scalability and flexibility in 5{G} mobile networks,''
  \emph{IEEE Communications magazine}, vol.~55, no.~5, pp. 72--79, 2017.

\bibitem{taleb2017multi}
T.~Taleb, K.~Samdanis, B.~Mada, H.~Flinck, S.~Dutta, and D.~Sabella, ``On
  multi-access edge computing: A survey of the emerging 5{G} network edge cloud
  architecture and orchestration,'' \emph{IEEE Communications Surveys \&
  Tutorials}, vol.~19, no.~3, pp. 1657--1681, 2017.

\bibitem{marquez2018should}
C.~Marquez, M.~Gramaglia, M.~Fiore, A.~Banchs, and X.~Costa-Perez, ``How should
  {I} slice my network? a multi-service empirical evaluation of resource
  sharing efficiency,'' in \emph{Proceedings of the 24th Annual International
  Conference on Mobile Computing and Networking}, 2018, pp. 191--206.

\bibitem{d2019slice}
S.~D’Oro, F.~Restuccia, A.~Talamonti, and T.~Melodia, ``The slice is served:
  Enforcing radio access network slicing in virtualized 5{G} systems,'' in
  \emph{IEEE INFOCOM 2019-IEEE Conference on Computer Communications}.\hskip
  1em plus 0.5em minus 0.4em\relax IEEE, 2019, pp. 442--450.

\bibitem{bega2017optimising}
D.~Bega, M.~Gramaglia, A.~Banchs, V.~Sciancalepore, K.~Samdanis, and
  X.~Costa-Perez, ``Optimising 5{G} infrastructure markets: The business of
  network slicing,'' in \emph{IEEE INFOCOM 2017-IEEE Conference on Computer
  Communications}.\hskip 1em plus 0.5em minus 0.4em\relax IEEE, 2017, pp. 1--9.

\bibitem{chen2018optimized}
X.~Chen, H.~Zhang, C.~Wu, S.~Mao, Y.~Ji, and M.~Bennis, ``Optimized computation
  offloading performance in virtual edge computing systems via deep
  reinforcement learning,'' \emph{IEEE Internet of Things Journal}, 2018.

\bibitem{sciancalepore2019rl}
V.~Sciancalepore, X.~Costa-Perez, and A.~Banchs, ``Rl-nsb: Reinforcement
  learning-based 5g network slice broker,'' \emph{IEEE/ACM Transactions on
  Networking}, vol.~27, no.~4, pp. 1543--1557, 2019.

\bibitem{hua2019gan}
Y.~Hua, R.~Li, Z.~Zhao, X.~Chen, and H.~Zhang, ``Gan-powered deep
  distributional reinforcement learning for resource management in network
  slicing,'' \emph{IEEE Journal on Selected Areas in Communications}, vol.~38,
  no.~2, pp. 334--349, 2019.

\bibitem{xu2021constrained}
Y.~Xu, Z.~Zhao, P.~Cheng, Z.~Chen, M.~Ding, B.~Vucetic, and Y.~Li,
  ``Constrained reinforcement learning for resource allocation in network
  slicing,'' \emph{IEEE Communications Letters}, vol.~25, no.~5, pp.
  1554--1558, 2021.

\bibitem{liu2020constrained}
Y.~Liu, J.~Ding, and X.~Liu, ``A constrained reinforcement learning based
  approach for network slicing,'' in \emph{2020 IEEE 28th International
  Conference on Network Protocols (ICNP)}.\hskip 1em plus 0.5em minus
  0.4em\relax IEEE, 2020, pp. 1--6.

\bibitem{liu2021resource}
------, ``Resource allocation method for network slicing using constrained
  reinforcement learning,'' in \emph{2021 IFIP Networking Conference (IFIP
  Networking)}.\hskip 1em plus 0.5em minus 0.4em\relax IEEE, 2021, pp. 1--3.

\bibitem{liu2021policy}
Y.~Liu, A.~Halev, and X.~Liu, ``Policy learning with constraints in model-free
  reinforcement learning: A survey,'' in \emph{Proceedings of the Thirtieth
  International Joint Conference on Artificial Intelligence}, 2021.

\bibitem{liu2021cts2}
Y.~Liu and X.~Liu, ``Cts2: Time series smoothing with constrained reinforcement
  learning,'' in \emph{Proceedings of the thirteenth Asian Conference on
  Machine Learning}, 2021.

\bibitem{chow2017risk}
Y.~Chow, M.~Ghavamzadeh, L.~Janson, and M.~Pavone, ``Risk-constrained
  reinforcement learning with percentile risk criteria,'' \emph{The Journal of
  Machine Learning Research}, vol.~18, no.~1, pp. 6070--6120, 2017.

\bibitem{tessler2018reward}
C.~Tessler, D.~J. Mankowitz, and S.~Mannor, ``Reward constrained policy
  optimization,'' \emph{International Conference on Learning Representations},
  2019.

\end{thebibliography}
